\documentclass[11pt,a4paper]{article}

\usepackage{graphicx}
\usepackage[english]{babel}
\usepackage{latexsym}
\usepackage{url}
\usepackage{amsfonts,amsmath,amsthm}
\usepackage{dsfont}
\usepackage{ifthen}
\usepackage{fancybox}
\usepackage{stmaryrd}
\usepackage{xspace}
\usepackage{enumerate}

\usepackage{color}
\definecolor{blu3}{rgb}{.1,.0,.4}
\usepackage{hyperref}
\hypersetup{colorlinks=true, linkcolor=blu3, urlcolor=blu3, citecolor=blu3, pdfpagemode=UseNone, pdfstartview=} 

\usepackage[margin=1.1in]{geometry}

\newtheorem{theorem}{Theorem}
\newtheorem{corollary}[theorem]{Corollary}
\newtheorem{lemma}[theorem]{Lemma}

\newtheorem{claim}[theorem]{Claim}

\newcounter{geometrical}
\renewcommand{\thegeometrical}{\Roman{geometrical}}
\newcounter{steps}

\newcommand{\RR}{\ensuremath{\mathbb R}}  
\newcommand{\ZZ}{\ensuremath{\mathbb Z}}  
\newcommand{\NN}{\ensuremath{\mathbb N}}  
\newcommand{\GG}{\ensuremath{\mathbb G}}  
\newcommand{\dconnected}{\textsc{DConnectivity}\xspace}
\newcommand{\connectivity}{\textsc{Connectivity}\xspace}

\newcommand{\PP}{\mathcal{P}}
\newcommand{\QQ}{\mathcal{Q}}
\newcommand{\SSS}{\mathcal{S}}
\newcommand{\UUU}{\mathcal{U}}
\newcommand{\CC}{\mathcal{C}}

\DeclareMathOperator{\MBST}{MBST}

\DeclareMathOperator{\OO}{\mathcal{O}}

\begin{document}

\title{Connectivity with Uncertainty Regions Given as Line Segments}

\author{Sergio Cabello\thanks{Faculty of Mathematics and Physics, University of Ljubljana, Slovenia, 
	and Institute of Mathematics, Physics and Mechanics, Slovenia.
	Email address: \texttt{sergio.cabello@fmf.uni-lj.si}.}
\and
	David Gajser\thanks{Faculty of Natural Sciences and Mathematics, University of Maribor, Slovenia, 
	and Institute of Mathematics, Physics and Mechanics, Slovenia. 
	Email address: \texttt{david.gajser@um.si}.}}

\maketitle

\begin{abstract}
	For a set $\QQ$ of points in the plane and a real number $\delta \ge 0$,
	let $\GG_\delta(\QQ)$ be the graph defined on $\QQ$
	by connecting each pair of points at distance at most $\delta$.
	
	We consider the connectivity of $\GG_\delta(\QQ)$ in the best scenario when the location of a few of the 
	points is uncertain, but we know for each uncertain point a line segment that contains it.
	More precisely, we consider the following optimization problem:
	given a set $\PP$ of $n-k$ points in the plane and a set $\SSS$ of $k$ line segments in the plane,
	find the minimum $\delta\ge 0$ with the property that we 
	can select one point $p_s\in s$ for each segment $s\in \SSS$ and  
	the corresponding graph $\GG_\delta ( \PP\cup \{ p_s\mid s\in \SSS\})$ is connected.
	It is known that the problem is NP-hard.
	We provide an algorithm to exactly compute an optimal solution in $\OO(f(k) n \log n)$ time, for 
	a computable function $f(\cdot)$. This implies that the problem is FPT when parameterized by $k$.
	The best previous algorithm uses $\OO((k!)^k k^{k+1}\cdot  n^{2k})$ time
	and computes the solution up to fixed precision. 
 
\end{abstract}

\section{Introduction}

For a set $\QQ$ of points in the plane and a real value $\delta\ge 0$, let $\GG_\delta(\QQ)$ be the graph
with vertex set $\QQ$ and edges connecting each pair of points $p,q$ at Euclidean distance
at most $\delta$. Connectivity of the graph $\GG_\delta(\QQ)$ is one of the basic properties
associated to the point set $\QQ$. For example, if the points represent devices that
can communicate and $\delta$ is the broadcasting range of each device, then the connectivity
of $\GG_\delta(\QQ)$ reflects whether all the devices form a connected network and they can
exchange information, possibly through intermediary devices.

In this work we consider the problem of finding the smallest $\delta$ 
such that $\GG_\delta(\QQ)$ is connected,
when some of the points from $\QQ$ are to be chosen from prescribed regions.
More precisely, we consider the following optimization problem.
\begin{quote}
	\connectivity \\
	Given a set $\UUU=\{U_1,\dots ,U_k\}$ of regions in the plane and 
	a set $\PP=\{p_{k+1},p_{k+2},\dots, p_n\}$ of points in the plane, 
	find
	\begin{align*}
			\delta^* ~=~ \min ~&~ \delta\\
						\mbox{s.t.}~&~ p_i \in U_i,\quad \mbox{ for } i=1,\dots, k\\
						&~ \GG_\delta (\{p_1,\dots,p_n\}) \mbox{ is connected}.
	\end{align*}
\end{quote}
In this work we will provide efficient algorithms for the \connectivity problem
when the regions are line segments and $k$ is small. 
See Figure~\ref{fig:example} for an example.

\begin{figure}[tb]
	\centering
	\includegraphics[width=\textwidth,page=1]{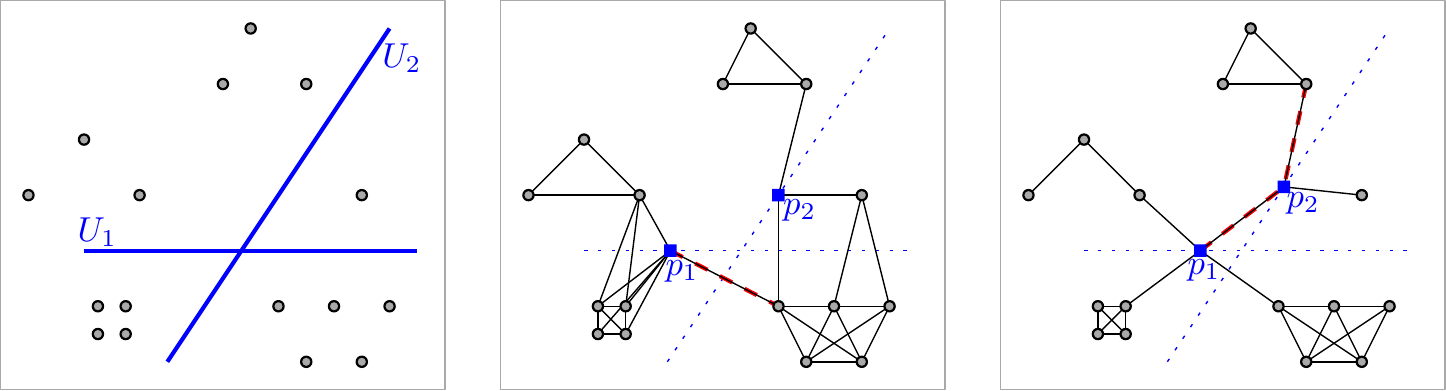}
	\caption{Example of the instances we consider, where the regions are line segments.
		Left: input data with two segments.
		Center and right: two possible choices of points $(p_1,p_2)\in U_1\times U_2$ 
			and the resulting graph $G_\delta$ for the minimum $\delta$ that makes $G_\delta$ connected.
			The edges of length $\delta$ are marked in dashed red; $\delta$ is different in each case.}
	\label{fig:example}
\end{figure}

There are other ways to characterize the connectivity of the graph $\GG_\delta(\QQ)$.
Let $D(p,r)$ denote the closed disk of radius $r$ centered at $p$.
Then, the graph $\GG_\delta(\QQ)$ is connected if and only if $\bigcup_{p\in \QQ} D(p,\delta/2)$ is 
a connected set.
Another characterization is provided by the Euclidean Minimum Bottleneck Spanning Tree of $\QQ$, denoted
by $\MBST(\QQ)$, a spanning tree of $\QQ$ where the length of the longest edge, called bottleneck edge, is minimized; 
a formal definition is given below.
The graph $\GG_\delta(\QQ)$ is connected if and only if $\MBST(Q)$ uses only
edges of length at most $\delta$.

It follows that the problem \connectivity is equivalent to the following problems:
\begin{itemize}
	\item choose a point $p_i$ per region $U_i$, where $i=1,\dots ,k$,
		in such a way that $\bigcup_{i=1,\dots,n} D(p_i,r)$ is connected and $r$ is the smallest possible;
		here the minimum $r$ is $\frac{\delta^*}{2}$.
	\item choose a point $p_i$ per region $U_i$, where $i=1,\dots ,k$, 
		in such a way that the $MBST$ on points $p_1,p_2,\ldots, p_n$ has shortest bottleneck edge.
\end{itemize} 

\paragraph{Related work.}
The problem we consider, \connectivity, was introduced by Chambers et al.~\cite{ChambersEFLSVSS17} under the
name of \emph{Best-Case Connectivity with Uncertainty}. In this setting each region $U_i$
is the uncertainty region for the point $p_i$.
They also considered the worst-case connectivity scenario, where one seeks for
the minimum $\delta$ such that $\GG_\delta (\{p_1,\dots,p_n\})$ is connected
for \emph{all} choices $p_i\in U_i$, where $i=1,\dots,k$. 
Thus, while in the best case we want to select points to achieve connectivity,
in the worst case we want to guarantee connectivity for all possible choices.

Chambers et al.~\cite{ChambersEFLSVSS17} showed that \connectivity is NP-hard even in 
the very restricted case when the uncertainty regions are vertical line segments of unit length
or when the uncertainty regions are axis-parallel unit squares.
For the case when the regions are line segments, they provide an algorithm that in
$\OO((k!)^k k^{k+1}\cdot (n+k)^{2k})$ time	computes an optimal solution up to fixed precision. 
The precision appears because of rounding the intermediary computations.

The case when the uncertainty regions are the whole plane, thus $U_1=\dots=U_k=\RR^2$
has been studied earlier under names like bottleneck Steiner tree problem
or bottleneck $k$-Steiner tree problem.
The results by Sarrafzadeh and Wong~\cite{SarrafzadehW92} imply that the problem is NP-hard.
Ganley and Salowe~\cite{GanleyS96} provided an approximation algorithm; they
also considered the rectilinear metric.
Wang and Li~\cite{WangL02} provided approximation algorithms, while
Wang and Du~\cite{WangD02} provided inapproximability results, assuming P$\neq$NP.
Bae et al.~\cite{blc-esebs-10} showed that the case of $k=1$ can be solved exactly
in $\OO(n\log n)$ time, and the case $k=2$ can be solved in $\OO(n^2)$ time.
Bae et al.~\cite{BaeCLT11} showed that the problem can be solved in 
$f(k)\cdot(n^k + n \log n)$ time, for some function $f(\cdot)$. This last
paper provides algorithms for the $L_1$ and $L_\infty$ metrics with a better running time, 
$f(k)\cdot(n \log^2 n)$.
Very recently, Bandyapadhyay et al.~\cite{BLLSX24} have shown that the problem 
can be solved in $f(k)\cdot n^{\OO(1)}$ time for some function $f(\cdot)$, 
which means that the bottleneck $k$-Steiner tree problem in the plane 
is fixed-parameter tractable with respect to $k$. The techniques can also
be used in other $L_p$ metrics.

Instead of minimizing the longest edge of the spanning tree (bottleneck version),
one could minimize the total length of the tree.
In the $k$-Steiner tree problem, we are given a set $\PP$ of points, and we
want compute a shortest Steiner tree for $\PP$ with at most $k$ Steiner points. 
This is similar to the \connectivity problem because we can take $U_1=\dots=U_k=\RR^2$,
but the optimization criteria is the sum of the lengths of the edges. 
The $1$-Steiner tree problem was solved in $\OO(n^2)$ time with the algorithm of 
Georgakopoulos and Papadimitriou~\cite{gp-1stp-87}.
Brazil et al.~\cite{BrazilRST15} solved the $k$-Steiner tree problem in $\OO(n^{2k})$ time.
Their technique can be adapted to the problem \connectivity and,
assuming that the uncertainty regions are convex of constant description size,
the problem reduces to $\OO(n^{2k})$ instances of quasiconvex programming \cite{Eppstein05},
each with $\OO(k)$ variables and constraints.

Closer to our setting is the work of Bose et al.~\cite{BoseDD22},
who considered the version of the $k$-Steiner tree problem 
where the points have to lie on given lines, 
and showed how to solve it in $\OO(n^k + n\log n)$ time.
Like the work of Brazil et al.~\cite{BrazilRST15}, their technique can be adapted to the problem 
\connectivity with uncertain segment regions. With this, the \connectivity
problem with uncertain segments reduces to $\OO(n^k)$ instances 
of quasiconvex programming \cite{Eppstein05}, each with $\OO(k)$ variables and constraints.

The problem \connectivity is an instance of the paradigm of computing optimal geometric
structures for points sets with uncertainty or imprecision:
each point is specified by a region in which the point may lie.
In this setting, we can model that the position of some points is certain by
taking the region to be a single point.
Usually one considers the maximum and the minimum value that can be attained for
some geometric structure,
such as the diameter, the area of the convex hull, or the minimum-area rectangle
that contains the points. This trend was started by L{\"{o}}ffler and van Kreveld;
see for example~\cite{LofflerK10a,LofflerK10b,KreveldL08} for some of the earlier works
in this paradigm.

\paragraph{Our results.}
We consider the \connectivity problem when each of the regions are line segments.
To emphasize this, we use $\SSS$ instead of $\UUU$ and $s_i$ instead of $U_i$.

First, we consider the decision version of the problem: 
Given a set $\SSS=\{s_1,\dots ,s_k\}$ of $k$ segments in the plane,
a set $\PP=\{p_{k+1},p_{k+2},\dots, p_n\}$ of $n-k$ points in the plane, 
and a value $\delta\ge 0$, decide whether there exist points
$p_i \in S_i$ for $i=1,\dots, k$
such that $\GG_\delta (\{p_1,\dots,p_n\})$ is connected.
We call this version of the problem \dconnected.

Our first main result is showing that \dconnected for segments can be solved using 
$f(k) n \log n$ operations, for some computable function $f(\cdot)$. 
In fact, after a preprocessing of the instance taking $\OO(k^2 n \log n)$ time,
we can solve \dconnected for any $\delta\ge 0$ using $f(k) n$ operations.
For this we use the following main ideas:
\begin{itemize}
	\item Instead of searching for connectivity, we search for a MBST.
	\item It suffices to restrict our attention to MBST of maximum degree $5$.
	\item The MBST will have at most $\OO(k)$ edges that are not part of the 
		minimum spanning tree of $\PP$.
	\item We can iterate over all the possible combinatorial ways
		how the uncertain points interact with the rest of the instance.
		Such interaction is encoded by a so-called topology tree with $\OO(k)$ nodes, 
		and there are $k^{\OO(k)}$ different options.
	\item For each topology tree $\tau$,
		we can employ a bottom-up dynamic programming across $\tau$
		to describe all the possible placements of points on the segment
		that are compatible with the subtree.
\end{itemize}

Our strongest result is showing that \connectivity for segments can be solved using 
$f(k) n \log n$ operations, for some other computable function $f(\cdot)$.
For this we use parametric search~\cite{Megiddo79,Megiddo83}, a generic
tool to transform an algorithm for the decision version of the problem into 
an algorithm for the optimization problem.
We provide a careful description of the challenges that appear when using
parametric search in our setting. 
While we can use Cole's technique~\cite{Cole87} in one of the steps,
we provide an alternative that is simpler, self-contained and uses
properties of the problem.
Eventually, we manage to solve the optimization problem without increasing
in the time complexity of the algorithm the dependency on $n$; 
the dependency on $k$ increases slightly.

Our result shows that \connectivity for line segments is fixed-parameter tractable when
parameterized by the number $k$ of segments.
The running time of our algorithms are a large improvement over the best previous time bound 
of $\OO(n^{2k})$ by Chambers et al.~\cite{ChambersEFLSVSS17} and the bound
of $\OO(n^k + n\log n)$ that could be obtained adapting the approach of
Bose et al.~\cite{BoseDD22}.

Compared to the work of Chambers et al.~\cite{ChambersEFLSVSS17}, we note that they
did not consider the decision problem, but instead guessed a \emph{critical}
path that defines the optimal value. Then they show that 
there are $\OO(n^{2k})$ critical paths.
Compared to the work of Bose et al.~\cite{BoseDD22} and Brazil et al.~\cite{BrazilRST15}, 
we first note that they are minimizing the sum of the length of the edges.
Optimizing the sum is usually harder than optimizing the bottleneck value.
Also, considering the decision problem is often useful for the bottleneck version because
it reduces the number of combinatorial options to consider, 
but this benefit is rarely present when minimizing the sum. 
To be more precise, in the decision version of our bottleneck problem, 
to extend a partial solution 
we only need to know whether it can be connected to a connected subgraph, but
we do not care to which vertex of the connected subgraph we are connecting. 
In contrast, when minimizing the sum of the lengths one has to carry the information of 
whom do you connect to and how much it costs. This means that we have to carry a
description of the cost function, which has a larger combinatorial description complexity.

When $k=O(1)$, our algorithms take $O(n \log n)$ time, which is asymptotically optimal:
since finding a maximum-gap in a set of $n$ unsorted numbers has a lower bound of 
$\Omega(n\log n)$ in the algebraic decision tree model, 
the problem of finding a MBST also takes $\Omega(n\log n)$, even without
uncertainty regions.

In our problem, we have to be careful about the
computability of the numbers appearing through the computation.
It is easy to note that we get a cascading effect of square roots. 
A straightforward approach is to round the numbers that appear through the computation
to a certain precision and bound the propagation of the errors. 
This is easy and practical, but then we do not get exact results.

The numbers computed through our algorithm have algebraic degree 
over the input numbers that depends on $k$, and thus can be manipulated exactly
if we assume that the input numbers are rational or of bounded algebraic degree.
The actual running time to manipulate these numbers exactly 
depends on $k$ and the assumptions on the input numbers, 
and this is hidden in the function $f(\cdot)$.
Below we provide background on algebraic numbers and computation trees, the
tool we use to manipulate numbers. 
Our running times are stated assuming exact computation.

To summarize, there are two sources that make the dependency on $k$ at least exponential: 
the number of topology trees considered in the algorithm is $k^{\OO(k)}$,
and we are manipulating $\Theta(n)$ numbers of algebraic degree at least $2^{\Omega(k)}$.

\paragraph{Organization.}
The rest of the paper is organized as follows.
In Section~\ref{sec:preliminary} we explain the notation and some of the concepts used in the paper.
We also provide some basic geometric observations.
In Section~\ref{sec:geometry} we provide a careful description of geometric operations
that will be used in the algorithm. We pay attention to the details to carry out the 
algebraic degree of the operations.
In Section~\ref{sec:decision} we provide the algorithm for the decision problem.
In Section~\ref{sec:root_functions} we analyze a function that will appear when analyzing parametric search;
we need to bound the algebraic degree of certain equations that appear in the algorithm.
In Section~\ref{sec:param} we provide the optimization algorithm using the paradigm of
parametric search.
We conclude in Section~\ref{sec:conclusions}.

\section{Notation, numbers and preliminary results}
\label{sec:preliminary}

\subsection{Notation} 
For each positive integer $n$, we use $[n]=\{1,\dots, n\}$.
For each set $A$ and $t\in\NN$, we use $\binom{A}{t}$ to denote the set of all subsets of $A$ with $t$ elements.

All graphs in this paper will be undirected. Hence, each graph will be given as an ordered pair $G=(V,E)$, where $V$ is the set of its vertices and $E\subseteq\binom{V}{2}$ is the set of its edges. We also write $V(G)$ and $E(G)$ for the sets of vertices and edges of $G$, respectively. We use the notation $uv$ for the edge with vertices $u$ and $v$. The graph $G$ is \emph{edge-weighted}, if it is accompanied by a weight function $w:E\rightarrow\RR_{\ge 0}$. The \emph{weight} of an edge-weighted subgraph $H$, denoted $w(H)$, is the sum of weights of all of its edges.

If $\QQ$ is a set of points in the plane, then we denote by $K_\QQ$ the complete graph on vertices $\QQ$. Such a graph is naturally accompanied by a weight function on edges that we call \emph{edge length}. In this setting, the edge length of the edge $uv$, where $u,v\in \QQ$, is the Euclidean distance between $u$ and $v$, denoted by $d(u,v)$. We also write $|uv|=d(u,v)$.

Let $G=(V,E)$ be an edge-weighted graph. If $T$ is a spanning tree of $G$, a \emph{bottleneck edge} of $T$ is an edge in $T$ with largest weight. We call $T$ a \emph{minimum bottleneck spanning tree} or \emph{MBST} of $G$, if its bottleneck edge has the smallest weight among all bottleneck edges of spanning trees of $G$. A \emph{minimum spanning tree} or \emph{MST} of $G$ 
is a spanning tree of $G$ that has minimum weight, over all spanning trees of $G$.

We will not distinguish between points and their position vectors. For two sets $A$ and $B$ in the plane, their \emph{Minkowski sum} is $A\oplus B =\{ a+b\mid a\in A, b\in B\}$. Recall that $D(p,r)$ is the closed disk with center $p$ and radius $r$. Then $A\oplus D(0,r)$ is precisely the set of points at distance at most $r$ from some point of $A$. We call this set the $r$-neighborhood of $A$. Any segment in this paper will be a line segment. A segment between two points $a$ and $b$ will be denoted as $\overline{ab}$. For a point $a$, we will use $||a||$ to denote the distance from $a$ to the origin.

\subsection{Representation of numbers and algebraic operations}
Each number that will be computed through our algorithm will be obtained from previous numbers
by either one of the usual arithmetic operations (addition, subtraction, multiplication or division), 
by computing a square root or by solving a polynomial equation of degree at most $2^{\OO(k)}$.
Each number that is computed inside the algorithm has its \emph{computation tree}.
This is a rooted tree that has the just described operations as internal vertices and input values as leaves.  
Each vertex in the tree represents a number that is computed by the operation described in this vertex 
applied to its descendants. See Figure~\ref{fig:computation_tree}, left, for an example. 

\begin{figure}[tb]
	\centering
	\includegraphics[width=.95\textwidth,page=2]{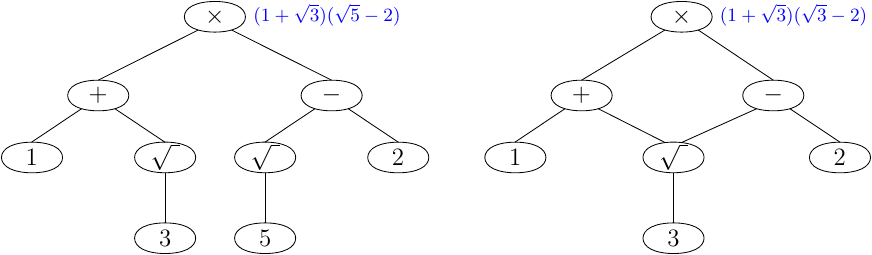}
	\caption{Left: Example of computation tree for $(1+\sqrt{3})(\sqrt{5}-2)$.
		Right: Example of a directed acyclic graph for computing $(1+\sqrt{3})(\sqrt{3}-2)$.}
	\label{fig:computation_tree}
\end{figure}

Computer algebra provides tools to manipulate and compare these numbers exactly; 
see for example~\cite[Section 10.4]{BPR06} or~\cite{Yap99}.
Exact computation is a paradigm promoted within Computational Geometry~\cite{BurnikelFMSS09,KettnerMPSY08,LiPY05,SY17}.

We will show that the depth of computation trees for numbers inside our algorithms will depend only on $k$. It follows that in our algorithms the time complexity of each of the numerical operations 
will always be a function only of $k$, independent of $n$. To avoid a cumbersome description,
\textbf{we will assume in the description that the manipulation of numbers takes $\OO(1)$ time}.
Let us note here that our algorithms will not compute numbers directly by their computation trees, rather by related rooted directed acyclic graphs which are essentially computation trees, but with joint nodes that produce the same number using the same operations. This is a known concept, used for example in~\cite{DAG}. To give an example, if we consider a slightly modified example from Figure~{\ref{fig:computation_tree}}, left, that computes $(1+\sqrt{3})(\sqrt{3}-2)$, we see that we do not need a node for number 5 and that we only need to compute $\sqrt{3}$ once, hence we need 2 fewer nodes. See Figure~{\ref{fig:computation_tree}}, right.
Note that such a transformation does not change the depth of the computation.

The following discussion is about bounding the algebraic degree of the numbers
appearing in the algorithm and it is aimed to readers familiar with algebraic computations.
In our final results we do not talk explicitly about the algebraic degree of the numbers
and it suffices to know that the problem of performing $\ell$ algebraic operations 
on at most $\ell$ numbers is decidable for any $\ell$.

A possible way to represent such a number $\alpha$ is as a univariate
polynomial $P(x)$ with integer coefficients together
with an isolating interval $I$; this is an interval with the property
that $\alpha$ is the unique root of $P(x)$ inside $I$.
The minimum degree of the polynomial representing $\alpha$ is the \emph{algebraic}
degree of $\alpha$ (over the integers).
 
It is known that if $\alpha$ and $\beta\neq 0$ are numbers of algebraic degree $k$, then
$\alpha\pm\beta$, $\alpha\cdot \beta$, $\alpha/\beta$ and $\sqrt{\alpha}$ 
have algebraic degree at most $k^2$.
Moreover, if $\alpha$ is the root of a polynomial of degree $d$ with coefficients
of algebraic degree $k$, then $\alpha$ has algebraic degree $\OO(d k^d)$.
To see this, we first construct a common field extension for all the coefficients, 
which will have degree $\OO(k^d)$, and then use the relation of the degree between towers
of field extensions.

In our decision algorithm, the numbers used in our computations have 
a computation tree of depth $\OO(k)$ on the input numbers, with internal nodes containing only arithmetic operations and square roots.
Therefore, we employ numbers of algebraic degree $2^{\OO(k)}$.
For the optimization problem, at the leaves of the computation tree
of some numbers we will also have a root of a polynomial of degree $2^{\OO(k)}$.
For a computation tree, it is always the same root of a polynomial that is being used.
Therefore, for the computation, it suffices to work
with an extension field of all the input numbers, which has degree $2^{\OO(k)}$ because
there is a single number of algebraic degree $2^{\OO(k)}$.
Therefore, also in the optimization problem, 
the numbers involved in the computation have algebraic degree $2^{\OO(k)}$.

\subsection{Properties of minimum trees}
In this section we present some well known claims that will be used later. 
The first property is a standard consequence of Kruskal's algorithm to compute
the MST.
\begin{claim}
\label{MBST_MST} 
	In any edge-weighted connected graph, all MSTs have the same weight of the bottleneck edge 
	which is the same as the weight of a bottleneck edge in any MBST.
	In particular, each MST is also a MBST.	
\end{claim}

The following result is also well known; see for example~\cite{MST5}.

\begin{claim}
\label{degree} 
	For each non-empty set $\QQ$ of points in the plane, 
	there exists a MST and a MBST of the complete graph $K_\QQ$ with maximum degree at most 5.
\end{claim}

\section{Geometric computations}
\label{sec:geometry}

In this section we describe representations of geometric objects that we will use 
and we present some basic geometric computations that will be needed in the algorithms. 

Each line segment $s$ is determined by a quadruple $(p_s, e_s, a_s, b_s)$, 
where $p_s$ is an arbitrary point on the line supporting $s$, 
$e_s$ is a unit direction vector of $s$, and real numbers $a_s\le b_s$ determine the endpoints of $s$, 
which are $p_s+a_s e_s$ and $p_s+b_s e_s$. 
When $a_s=b_s$, the segment degenerates to a single point; we keep calling it a segment to avoid
case distinction. 
We will write $s=(p_s, e_s, a_s , b_s)$. 
We could easily allow for non-unit vectors $e_s$, but then some of the equations below
become a bit more cumbersome.
We will use such representations of line segments because we will often consider their subsegments, 
hence $p_s$ and $e_s$ will remain constant and only $a_s$ and $b_s$ will change. 

A \emph{segmentation} on a line $L$ denotes a union of pairwise-disjoint line segments on $L$. 
Some of the segments may be a single point.
If the line is given by a position vector of some point $p$ on the line and a unit direction vector $e$, 
then we will represent a segmentation of this line with $N$ disjoint subsegments and points as a $(2N+2)$-tuple 
\[
	X = (p,e,a_1,b_1,a_2,b_2,\ldots,a_N,b_N),
\]
where $a_i\leq b_i$, for each $i\in[N]$, and $b_i< a_{i+1}$, for each $i\in[N-1]$. 
For each $i\in[N]$, the line segment $(p,e,a_i,b_i)$ is part of the segmentation $X$. 
We call $N$ the size of the segmentation $X$.
See Figure~\ref{fig:segmentation} for an example.

\begin{figure}[tb]
	\centering
	\includegraphics[scale=1.2,page=3]{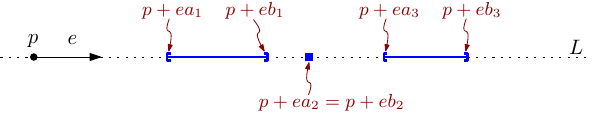}
	\caption{A segmentation of size 3 where the second segment degenerates to a point.}
	\label{fig:segmentation}
\end{figure}

\refstepcounter{geometrical}
\paragraph{\thegeometrical) Intersection of two lines.}\label{geom:lineline} 
Suppose we are given two lines $L_1\equiv p_1+ t_1 e_1$ and $L_2\equiv p_2+ t_2e_2$, 
where $p_1$ and $p_2$ are points on the first and second line, respectively, 
$e_1$ and $e_2$ are their unit direction vectors, respectively, and $t_1,t_2\in \RR$ are parameters. 
We would like to compute $L_1\cap L_2$.

If the lines are parallel, which is equivalent to $e_1=\pm e_2$, then we have two options. 
If the system of two linear equations in one unknown $p_2=p_1+ te_1$ has some solution, 
then the lines are equal and $L_1\cap L_2=L_1$. Otherwise, $L_1\cap L_2=\emptyset$.

If the lines are not parallel, then the system of two linear equations with two unknowns 
$p_1+ t_1e_1=p_2+ t_2e_2$ has a unique solution $(t_1^*,t_2^*)$. 
In this case $L_1\cap L_2$ contains exactly one point, namely $p_1+ t_1^*e_1=p_2+ t_2^*e_2$.

In all cases, we can compute $L_1\cap L_2$ with $\OO(1)$ arithmetic operations.

\refstepcounter{geometrical}
\paragraph{\thegeometrical) Intersection of a circle with a line.}\label{geom:circleline} 
Suppose we are given a circle (curve) $C$ with center at $c$ and radius $\delta>0$, and a line $L\equiv p+ te$, 
where $p$ is a point on the line, $e$ is its unit direction vector and $t\in \RR$ is a parameter.
We would like to compute the intersection $C\cap L$. 

This means that we need to solve the equation $||p+te-c||=\delta,$ which is quadratic in the unknown $t\in \RR$. 
This is equivalent to solving
\[ 
	||p+te-c||^2 =\delta^2, 
\]
which can be rewritten as 
\[ 
	t^2+2te\cdot(p-c)+||p-c||^2-\delta^2=0.
\]
 
Let $\Delta=4\left(e\cdot(p-c)\right)^2-4||p-c||^2+4\delta^2$ be the discriminant of this equation in $t$. 
We consider the following 3 cases.
\begin{enumerate}[i)]
    \item \label{D<0} $\Delta<0$. In this case, $C\cap L=\emptyset$.
    \item \label{D=0} $\Delta=0$. In this case, there is one solution of our quadratic equation, which is $t_0=e\cdot(c-p)$. Hence, $C\cap L=\{p+ t_0e\}$. The line $L$ is tangent to $C$.
    \item \label{D>0}$\Delta>0$. In this case, there are two solutions of our quadratic equation, which are 
    \[
        t_1=e\cdot(c-p)-\frac{1}{2}\sqrt{\Delta},~~~~
        t_2=e\cdot(c-p)+\frac{1}{2}\sqrt{\Delta},
    \]
    hence $C\cap L=\{p+ t_1e,p+ t_2e\}$.
\end{enumerate}

\refstepcounter{geometrical}
 \paragraph{\thegeometrical) Intersection of a disk with a line segment.}\label{geom:disksegment} 
Suppose we are given a disk $D$ with center at $c$ and radius $\delta> 0$, and a line segment $s=(p_s, e_s, a_s, b_s)$.
We would like to compute the intersection $D\cap s$. 

First we compute the intersection between the boundary $C$ of $D$ and the line $L$ 
that contains $s$, as described in~\ref{geom:circleline}).
 
\begin{enumerate}[i)]
    \item \label{noIntersection} If $C\cap L=\emptyset$, then $D\cap s=\emptyset$.
    \item \label{touching} If $L$ is tangent to $C$ at the point $p_s+ t_0e_s$, for some $t_0$, 
		then we verify whether $t_0$ is between $a_s$ and $b_s$. 
		In this case $D\cap s$ contains exactly the point $p_s+ t_0e_s$, otherwise $D\cap s=\emptyset$.
    \item \label{intersecting} If $L$ intersects $C$ in two points $p_s+ t_1 e_s$ and $p_s+ t_2e_s$, 
		where $t_1<t_2$, this means that $D\cap L$ is the line segment $s'=(p_s,e_s,t_1,t_2)$. 
		Hence, $D\cap s=s'\cap s$ and can be computed with $\OO(1)$ additional comparisons between
		the values $t_1,t_2,a_s,b_s$.
\end{enumerate}
Note that to compute $D\cap s$ all arithmetic computations were performed when computing $C\cap L$. 
The rest of the operations are only $\OO(1)$ comparisons.

\refstepcounter{geometrical}
\paragraph{\thegeometrical) Intersection of two segmentations of a line.} \label{geom:segmentations} 
Suppose we have segmentations $X_1$ and $X_2$ on some line of sizes $N_1$ and $N_2$, respectively.
We would like to compute the segmentation $X_1\cap X_2$. 
The segmentation $X_1\cap X_2$ has size at most $N_1+N_2$ because the 
points of the segmentation $X_1\cap X_2$ are points of $X_1$ or $X_2$. 
Moreover, because each segmentation already has its segments and points ordered, 
we can compute $X_1\cap X_2$ with $\OO(N_1+N_2)$ comparisons.
The idea is similar to the merging of two sorted lists. For example, 
we can use 3 pointers, one for $X_1$, one for $X_2$ and one for the merged list, 
that are traversed simultaneously once and, for each two consecutive points in the merged list, 
we verify whether the corresponding line segment is a subset of $X_1\cap X_2$ or not in $\OO(1)$ time.

\refstepcounter{geometrical}
\paragraph{\thegeometrical) Computing a Voronoi diagram on a line segment for some set of points.} \label{geom:voronoi}
Suppose we are given a set $\QQ=\{q_1,q_2,\ldots, q_N\}$ of $N$ points in the plane 
and a line segment $s=(p_s, e_s, a_s, b_s)$.
We would like to compute the \emph{Voronoi diagram on the line segment $s$ for points in $\QQ$}.
See Figure~\ref{fig:Voronoi}.
For our purpose we can say that such a Voronoi diagram
is a sequence of pairs $(q'_1, J_1),(q'_2, J_2),\ldots, (q'_{N'}, J_{N'})$ with the following properties: 
\begin{itemize}
    \item for each $i\in [N']$, the point $q'_i$ belongs to $\QQ$ and 
		$J_i$ is a segment contained in $s$,
    \item for each $i\in [N']$ and for each point $p$ in $J_i$, 
		the smallest distance from $p$ to any point in $\QQ$ is $d(p,q'_i)$,		
	\item the union of the segments $J_1,\ldots, J_{N'}$ is the segment $s$,
    \item for each $i\in [N'-1]$, $J_i$ and $J_{i+1}$ have exactly a boundary point in common, 
    \item for each distinct and non-consecutive $i,j\in [N']$, the segments $J_i$ and $J_j$ are disjoint, and
    \item for each $q\in \QQ$ there is at most one $i\in [N']$ with $q_i'=q$.
\end{itemize}

\begin{figure}[tb]
	\centering
	\includegraphics[width=\textwidth,page=4]{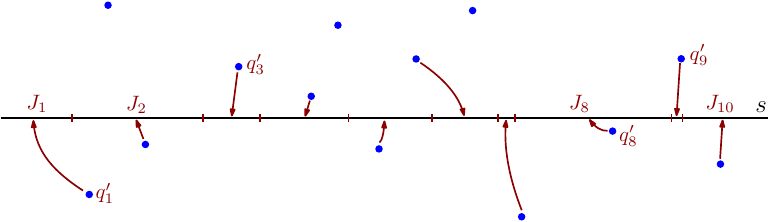}
	\caption{Example of Voronoi diagram on a segment $s$ defined by points.
		For each Voronoi cell we mark its closest site.}
	\label{fig:Voronoi}
\end{figure}

It is well known that this can be done in $\OO(N\log N)$ time by computing 
the Voronoi diagram of $\QQ$ in $\RR^2$ and intersecting it with the segment $s$.
See for example the textbook by de Berg et al.~\cite[Chapters 7 and 2]{BergCKO08}.

Because we are only interested in the segment $s$, this can be done also easily using
a simple divide-and-conquer approach, as follows.
When $\QQ$ contains a single point, its Voronoi diagram is the whole segment $s$.
When $\QQ$ has at least two points, we split $\QQ$ arbitrarily into two sets, $\QQ_1$ and $\QQ_2$, 
of roughly the same size, and recursively compute the Voronoi diagram on $s$ for $\QQ_1$ and for $\QQ_2$.
We get those two Voronoi diagrams as sequences
$(q'_1, J'_1),(q'_2, J'_2),\ldots, (q'_{N'}, J'_{N'})$ for $\QQ_1$
and 
$(q''_1, J''_1),(q''_2, J''_2),\ldots, (q''_{N''}, J''_{N''})$ for $\QQ_2$.
To merge them, we first compute all the non-empty intervals $J'_i\cap J''_j$ for $i\in [N']$ and $j\in [N'']$.
This takes $O(N'+N'')$ time because the two sequences are sorted along $s$.
We also obtain the output sorted along $s$.
We know that, for each such non-empty $J'_i\cap J''_j$, the closest point of $\QQ$ is either $q'_i$
or $q''_j$.
For each such non-empty $J'_i\cap J''_j$, we compute the intersection $p_{i,j}$ of the bisector for 
$q'_i$ and $q''_j$ with the line supporting $s$. If the intersection point $p_{i,j}$ lies on $J'_i\cap J''_j$,
we split $J'_i\cap J''_j$ into two intervals at $p_{i,j}$. 
We have obtained a sequence of intervals along $s$ with the property that each point in an interval
has the same closest point of $\QQ$. With a final walk along the intervals,
we merge adjacent intervals with the same closest point in $\QQ$ into a single interval.
This merging step takes $O(N'+N'')$ time.

If $T(N)$ denotes the running time of the divide-and-conquer algorithm for $N$ points, 
we have the recurrence $T(N)= \OO(N)+ T(|\QQ_1|) +T(|\QQ_2|)$, with base case $T(1)=\OO(1)$.
Because $|\QQ_1|$ and $|\QQ_2|$ are approximately $N/2$, 
this solves to $T(N)=\OO(N \log N)$.

Each number computed in the procedure is obtained from the input data
by $\OO(1)$ additions, subtractions, multiplications and divisions. 
This is because each value is obtained by computing the
intersection of a bisector of two points of $\QQ$ with the line supporting $s$.

\refstepcounter{geometrical}
\paragraph{\thegeometrical) Intersection of a union of disks with a line segment, equipped with a Voronoi diagram.} \label{geom:unionsegment} 
Suppose we are given $N$ disks $D(q_1,\delta),D(q_2,\delta),\ldots, D(q_N,\delta)$ with radius $\delta>0$, a line segment $s=(p_s, e_s, a_s, b_s)$ and a Voronoi diagram $(q'_1, J_1),(q'_2, J_2),\ldots, (q'_{N'}, J_{N'})$ on $s$ for the points $q_1,q_2,\ldots, q_{N}$.
We would like to compute the intersection 
\[
	X=\big(\bigcup_{j\in [N]} D(q_j,\delta)\big)\cap s.
\]
Although we do not need the Voronoi diagram to compute $X$, we will use it to compute it in a linear number of steps.
We first observe that (see Figure~\ref{fig:unionsegment})
\[
	X=\bigcup_{j\in [N']} \big(D(q'_j,\delta)\cap J_j\big).
\]
This implies that $X$ can be computed in time $\OO(N)$ by applying $N'$ times the procedure in~\ref{geom:disksegment}) and then joining each two consecutive line segments $D(q'_j,\delta)\cap J_j$ and $D(q'_{j+1},\delta)\cap J_{j+1}$, if they have a common endpoint $J_j\cap J_{j+1}$.
Note that the output is a segmentation contained in $s$.

\begin{figure}[tb]
	\centering
	\includegraphics[width=\textwidth,page=5]{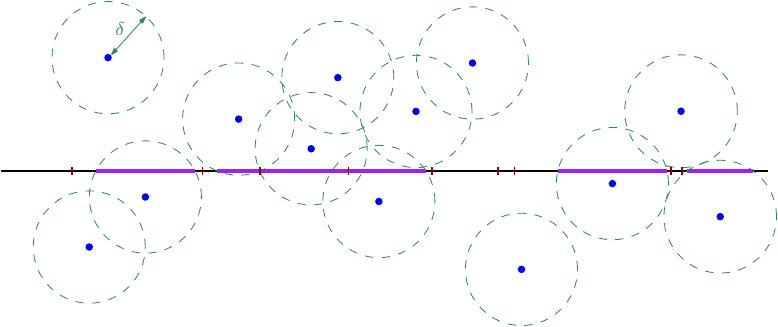}
	\caption{Computing the intersection of a segment and the union of congruent disks using the Voronoi diagram
		of the centers.}
	\label{fig:unionsegment}
\end{figure}

\refstepcounter{geometrical}
\paragraph{\thegeometrical) Intersection of a $\delta$-neighborhood of a segmentation with a line segment.} \label{geom:segmentationsegment} 
Suppose we are given some $\delta>0$, a line segmentation $X=(p_X, e_X, a_1, b_1, \ldots , a_N, b_N)$ 
and a line segment $s=(p_s, e_s, a_s, b_s)$. We would like to compute the intersection 
$X'=\left(X \oplus D(0,\delta)\right)\cap s$ between $s$ and the $\delta$-neighborhood of $X$.
This is a segmentation in the line supporting the segment $s$. 
See the top of Figure~\ref{fig:segmentationsegment}.

\begin{figure}[tb]
	\centering
	\includegraphics[scale=1.1,page=6]{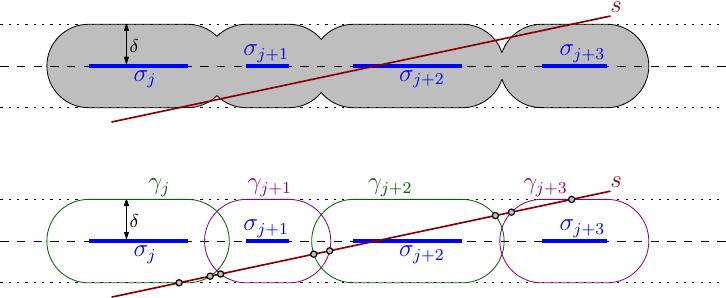}
	\caption{Top: an example showing the input in~\ref{geom:segmentationsegment}).
		Bottom: the curves $\gamma_i$ and their intersection with the segment $s$.}
	\label{fig:segmentationsegment}
\end{figure}

For each $j\in[N]$, let $\sigma_j$ be the $j$-th segment of the segmentation $X$;
thus $\sigma_j=(p_X,e_X,a_j,b_j)$.
For each $j\in[N]$, let $\gamma_j$ be the boundary of $\sigma_j\oplus D(0,\delta)$;
see the bottom of Figure~\ref{fig:segmentationsegment}.
The boundary of $\gamma_j$ consists of two semicircles of radius $\delta$, 
one centered at $p_X+a_j e_X$ and another centered at $p_X+b_j e_X$, 
and two copies of the segment $\sigma_j$
translated perpendicularly to $\sigma_j$ by $\delta$, one in each direction.
 
For each $j\in[N]$, we compute the (possibly empty, possibly degenerate) segment 
$\eta_j = \bigl( \sigma_j\oplus D(0,\delta)\bigr)\cap s$. 
See the top of Figure~\ref{fig:segmentationsegment2}.
To do this we compute $\gamma_j\cap s$ using~\ref{geom:lineline}) and~\ref{geom:circleline}) 
for the lines and circles supporting pieces of $\gamma_j$ and,
for each intersection point we find, we test whether it indeed belongs to $\gamma_j$.
We also test whether the endpoints of $s$ belong to $\sigma_j\oplus D(0,\delta)$,
as the segment $s$ may start inside multiple regions $\sigma_j\oplus D(0,\delta)$.
This takes $\OO(N)$ time and each number we computed requires $\OO(1)$ arithmetic
operations, square roots, and comparisons. (The explicit computation of $\eta_j$ may require
the square root.)
Each such non-empty segment $\eta_j$ is represented as $\eta_j = (p_s,e_s,a'_j,b'_j)$,
using the same point $p_s$ and direction unit vector $e_s$ for all $j\in[N]$.

\begin{figure}[tb]
	\centering
	\includegraphics[scale=1.1,page=8]{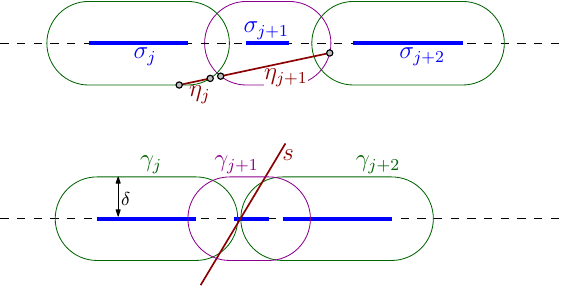}
	\caption{Top: an example showing $\eta_j$ and $\eta_{j+1}$. Note that $\eta_{j+2}$ 
		is not shown, but it would overlap with $\eta_{j+1}$.
		Bottom: example showing that $s$ may not enter $\gamma_j,\gamma_{j+1},\gamma_{j+2}$
		in that order.}
	\label{fig:segmentationsegment2}
\end{figure}

If we found no intersections, meaning that $\eta_j = \emptyset$ for all $j\in [N]$,
we return $X'=\emptyset$. Otherwise, the segments $\eta_1,\ldots,\eta_N$ may overlap and we have
to merge them. 
We must be careful because the line segment $s$ does not need to enter the regions $\gamma_j$, $\gamma_{j+1}$ and $\gamma_{j+2 }$, for some $j\in[N-2]$ in this order. 
See for example the bottom of Figure~\ref{fig:segmentationsegment2}.
However, we can be sure about two things:
\begin{itemize}
	\item If $\eta_i\neq \emptyset$ and $\eta_j\neq \emptyset$ for some $1\leq i<j\leq N$, then $\eta_k\neq \emptyset$ for all $k\in\{i,i+1,\ldots,j\}$.
	\item If $\eta_i\cap\eta_j\neq \emptyset$ for some $1\leq i<j\leq N$, then $\eta_i\cap\eta_k\neq \emptyset$  for all $k\in\{i,i+1,\ldots,j\}$.
\end{itemize}

This means that we can merge the segments $\eta_1,\ldots,\eta_N$ by considering only segments with adjacent indices.
One way to do it is to compute 
\[ 
	m =\min \{ j\in [N]\mid \eta_j\neq \emptyset\} ~~\text{ and }~~ M=\max \{ j\in [N]\mid \eta_j\neq \emptyset\},
\]
and make a linked list with the collinear segments 
\[
	\eta_m = (p_s,e_s,a'_m,b'_m), \ldots, \eta_M = (p_s,e_s,a'_M,b'_M),
\]
in that order. Walking along the list, whenever two consecutive segments $\eta = (p_s,e_s,a,b)$ and $\eta' = (p_s,e_s,a',b')$ in
the list intersect, which can be checked by sorting numbers $a,a',b$ and $b'$, we merge them into the single segment 
$\eta''$, which replaces $\eta$ and $\eta'$ in the list. 
If the final list is $\tilde\eta_1 = (p_s,e_s,\tilde a_1,\tilde b_1),\ldots, \tilde\eta_j = (p_s,e_s,\tilde a_{N'},\tilde b_{N'})$,
we then have
\[
	\left(X\oplus D(0,\delta)\right)\cap L_s = (p_s,e_s,\tilde a_1,\tilde b_1, \ldots, \tilde a_{N'},\tilde b_{N'}).
\]
The whole computation takes $\OO(N)$ time and 
each number in the output is obtained from the input data by performing $\OO(1)$ arithmetic operations 
and square roots.

\section{Solving the decision version}
\label{sec:decision}

In this section we show how to solve the decision of the problem, \dconnected,
when the uncertain regions are segments. 
Let us recall the problem.
Given a set $\SSS=\{s_1,\dots ,s_k\}$ of $k$ segments in the plane,
a set $\PP=\{p_{k+1},p_{k+2},\dots, p_n\}$ of $n-k$ points in the plane, 
and a value $\delta\ge 0$, decide whether there exist points
$p_i \in s_i$ for $i=1,\dots, k$
such that $\GG_\delta (\{p_1,\dots,p_n\})$ is connected.
Denoting by $\delta^*$ the optimal value in the optimization problem, 
we want to decide whether $\delta\ge \delta^*$.

We will consider the case $\delta=0$ separately. Then we will present an algorithm that solves \dconnected for $\delta>0$ in 4 parts, which are in bold in the next sentence. \textbf{DStep~\ref{step:build_tree}} will be executed first, followed by \textbf{DStep~\ref{step:Voronoi}}. Then we will use a \textbf{DLoop} inside of which \textbf{DStep~\ref{step:inside_loop}} will be performed $k^{\OO(k)}$ times. The letter \textbf{D} in \textbf{DStep} and \textbf{DLoop} specifies that the steps are of the algorithm for the decision variant.

\paragraph{Case $\delta=0$.} Solving the problem \dconnected for $\delta=0$ is equivalent to deciding whether all line segments from $\SSS$ have a common point which is the same as each point in $\PP$. This can clearly be done in time $\OO(n)$. For the rest of the description, we assume $\delta>0$.

\refstepcounter{steps}
\paragraph{DStep \thesteps.}\label{step:build_tree} We compute a MST $T$ for points in $\PP$.
It is well known that the tree $T$ can be computed in $\OO(n\log n)$ time~\cite{MSTtime}.
The most usual way is noting that a MST of any Delaunay triangulation of the point set $\PP$
is a MST for $\PP$~\cite[Exercise 9.11]{BergCKO08}. 

We remove all edges from $T$ that are longer than $\delta$. Let the remaining connected components of the tree $T$ be $\CC=\{C_1,C_2,\ldots, C_{\ell}\}$. Note that we removed exactly $\ell-1$ edges. If $\ell>4k+1$, return \textsl{FALSE}.
We first show that this decision based on $\ell>4k+1$ is correct.

\begin{lemma}
\label{components}
	Any two points in $\PP$ from distinct components $C_1,C_2,\ldots, C_{\ell}$ are more than $\delta$ apart.
\end{lemma}
\begin{proof}
	If the lemma was false, then there would exist points $p_a,p_b\in \PP$ such that 
	$d(p_a,p_b)\leq \delta$ and $p_a\in C_{a'}$, $p_b\in C_{b'}$, where $a'\neq b'$. 
	Clearly, the edge $p_a p_b$ is not part of the tree $T$. 
	If we add it to the tree, we get a cycle that connects the points $p_a$ and $p_b$ 
	either via the edge $p_a p_b$ or via a path that contains an edge $e$ that was longer than $\delta$, 
	because the components $C_{a'}$ and $C_{b'}$ are distinct. 
	It follows that if we add the edge $p_a p_b$ and remove the edge $e$ from $T$, 
	we get a spanning tree on points $\PP$ with weight less than the weight of $T$, 
	which is a contradiction. Hence, the lemma must be true.
\end{proof}

\begin{lemma}
\label{ell_bound}
	If $\ell> 4k+1$, then $\delta<\delta^*$.
\end{lemma}
\begin{proof}
	We show the contrapositive statement. Hence, we assume $\delta\geq\delta^*$.
	Let $p_i \in s_i$, for $i\in[k]$, be such points that any MBST on points $p_1,p_2,\ldots, p_n$ 
	has a bottleneck edge of length $\delta^*$. Such points exist by definition of $\delta^*$. 
	By Claim~\ref{MBST_MST}, any MST on these points also has a bottleneck edge of length $\delta^*$. 
	Let $T'$ be a MST tree on points $p_1,p_2,\ldots, p_n$ with the degree of each vertex at most 5. 
	Such a tree exists by Claim~\ref{degree}. Hence, there are at most $5k$ neighbors of points $p_1,p_2,\ldots, p_k$. 
	By Lemma~\ref{components}, any two points from distinct components $C_1,C_2,\ldots, C_{\ell}$ 
	are not connected directly with edges of $T'$, hence for each component there is an edge in $T'$ 
	from a point in this component to a point in $\{p_1,p_2,\ldots, p_k\}$.
	Let $E_1$ be the set of edges in $T'$ that have exactly one vertex in $\{p_1,p_2,\ldots, p_k\}$ 
	and let $E_2$ be the set of edges in $T'$ that have both vertices from the set $\{p_1,p_2,\ldots, p_k\}$. 
	We have $|E_1|+2|E_2|\leq 5k$. Because there are no edges between components $C_i$ and $C_j$, for $i\neq j$, 
	we have $|E_1|+|E_2|\geq k+\ell-1$. This is because using edges from $E_1\cup E_2$ 
	we have to connect at least $k+\ell$ distinct ``clusters'' of points: each point in $\{p_1,p_2,\ldots, p_k\}$ 
	is one ``cluster'' and each component $C_1,C_2,\ldots, C_{\ell}$ contributes at least one cluster. 
	This implies $k+\ell-1\leq 5k$, which gives $\ell\leq 4k+1$.
\end{proof}

\refstepcounter{steps}
\paragraph{DStep~\thesteps.}\label{step:Voronoi} For each component $C_i\in\CC$ and for each line segment $s_j\in\SSS$, we compute the 
Voronoi diagram on $s_j$ of the points in $C_i$. This can be done with $\OO(kn\log n)$ steps as explained in~\ref{geom:voronoi}). 

\paragraph{DLoop.} We treat each line segment from $\SSS$ and each component from $\CC$ as an abstract node and we iterate over all possible trees on these $k+\ell$ nodes such that
\begin{enumerate}[a)]
    \item each node from $\SSS$ has degree at most 5 and
    \item no two nodes from $\CC$ are adjacent.
\end{enumerate}
We call each such a tree a \emph{topology tree}, because it describes a potential way to connect the components 
in $\CC$ via points from the line segments in $\SSS$. 
See Figure~\ref{fig:topologytree} for an example.
Note that we are reserving the term \emph{node}
for each connected component of $\CC$ and each segment of $\SSS$. In this way we distinguish nodes from a topology
tree from vertices of other graphs.
   
\begin{figure}[tb]
	\centering
	\includegraphics[width=\textwidth,page=7]{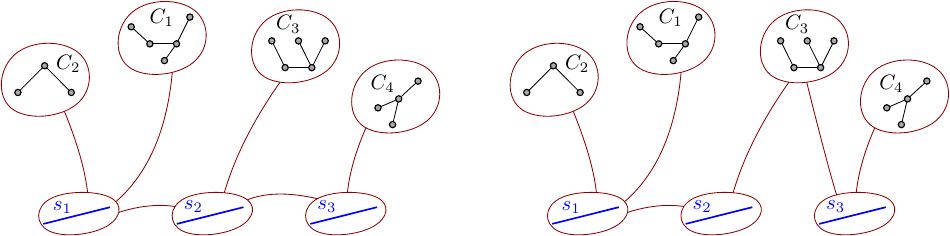}
	\caption{Two examples of topology trees.  The left example has one 
		significant topology subtree (a concept we will introduce later),
		while the right one has two significant topology subtrees, one spanned by the nodes
		$\{s_1,s_2,C_1,C_2,C_3\}$ and one spanned by $\{s_3,C_3,C_4\}$.}
	\label{fig:topologytree}
\end{figure}

We say that a topology tree $\tau$ is \emph{$\delta$-realizable}, if there exist points $p_j\in s_j$, 
for each $j\in [k]$, such that
\begin{enumerate}[a)]
    \item for each edge $s_i s_j$ in $\tau$, where $s_i,s_j\in \SSS$, it holds $d(p_i,p_j)\leq\delta$, and
    \item for each edge $C_i s_j$ in $\tau$, where $C_i\in \CC$ and $s_j\in \SSS$, there exists a point $p\in C_i$ 
		such that $d(p,p_j)\leq\delta$. (Note that we may use different points $p\in C_i$ for different 
		edges $C_i s_j$, $C_i s'_j$ of $\tau$.)
\end{enumerate}

\begin{lemma}
\label{topology}
    There exists a $\delta$-realizable topology tree if and only if  $\delta\geq\delta^*$.
\end{lemma}
\begin{proof}
	Let $\tau$ be a $\delta$-realizable topology tree. Let us fix points $p_j\in s_j$, for each $j\in[k]$, 
	such that
	\begin{enumerate}[\indent a)]
		\item for each edge $s_i s_j$ in $\tau$, where $s_i,s_j\in \SSS$, it holds $d(p_i,p_j)\leq\delta$ and
		\item for each edge $C_i s_j$ in $\tau$, where $C_i\in \CC$ and $s_j\in \SSS$, 
			there exists a point $p\in C_i$ such that $d(p,p_j)\leq\delta$. 
	\end{enumerate}
	Let $G=\GG_\delta (\{p_1,p_2,\ldots p_n\})$. 
	We will prove that $G$ is connected, 
	which by definition of $\delta^*$ implies $\delta\geq\delta^*$. 
	Note that any two vertices in the same component $C$ of $\CC$ are connected in $G$.
	Indeed, since there exists a path connecting two vertices of $C$ in the MST $T$ 
	that uses only edges of length at most $\delta$, such a path is also present in $G$.
	
	Consider two arbitrary points $p_i$ and $p_j$, where $i,j\in[n]$. Let $v_i$ be the node of $\tau$
	that contains $p_i$; it may be that $v_i=C$ for some $C\in \CC$ or that $v_i=s$ for 
	some $s\in \SSS$.
	Similarly, let $v_j$ be the node of $\tau$ that contains $p_j$.
	If $v_i=v_j$ and $v_i$ is a segment of $\SSS$, then $p_i=p_j$ and they are connected in $G$.
	If $v_i=v_j$ and $v_i$ is a connected component $C\in \CC$, then they are also connected in $G$.

	It remains to handle the case when $v_i\neq v_j$. Because $\tau$ is a tree,
	then there exist nodes $u_1,u_2,\ldots, u_m$ of $\tau$ 
	such that $u_1u_2\cdots u_m$ is a path in $\tau$ with $u_1=v_i$ and $u_m=v_j$ 
	We will construct a corresponding walk in $G$ that connects 
	$p_i$ and $p_j$ by transforming the path $u_1u_2u_3\cdots u_m$ in the following way.
	\begin{enumerate}
		\item If $u_1=s_i \in \SSS$, then we replace $u_1$ by $p_i$.
			Otherwise $u_1=C$ for some $C\in \CC$. 
			From the definition of $\tau$ it follows that the node $u_2=s_y$ for some segment $s_y\in \SSS$, 
			and there exists a point $p$ from $C$ such that $d(p,p_y)\leq\delta$. 
			We replace the node $u_1$ with a path from $p_i$ to $p$ in $C$. 
		\item If $u_m=s_j \in \SSS$, then we replace $u_m$ by $p_j$.
			Otherwise $u_m=C$ for some $C\in \CC$.
			From the definition of $\tau$ it follows that the node $u_{m-1}=s_y$ for some segment $s_y\in \SSS$, 
			and there exists a point $p$ from $C$ such that $d(p,p_y)\leq\delta$. 
			We replace the node $u_m$ with a path from $p$ to $p_j$ in $C$. 
		\item For each vertex $u_x\in\{u_2,\ldots, u_{m-1}\}$ from $\CC$, 
			we know from the definition of $\tau$ that $u_{x-1}=s_y\in \SSS$ and $u_{x+1}=s_z\in \SSS$, 
			for some $y,z\in[k]$. By definition of $\tau$ there exist points $q_x$ and $q'_x$ from the component
			$u_x\in\CC$ such that $d(q_x,p_y)\leq\delta$ and $d(q'_x,p_z)\leq\delta$. 
			We replace $u_x$ with a path from $q_x$ to $q'_x$ in the component $u_x$. 
		\item For each vertex $s_x\in\{u_2,u_3,\ldots, u_{m-1}\}$ from $\SSS$, 
			we replace $s_x$ with the corresponding point $p_x\in s_x$.
	\end{enumerate}
	It is clear that we transformed the path $u_1u_2u_3\cdots u_m$ in $\tau$ into a walk in the graph $G$ from $p_i$ to $p_j$. 
	Hence, $G$ is connected, which implies $\delta\geq\delta^*$.

	For the other direction, assume $\delta\geq\delta^*$. 
	Let points $p_i\in s_i$, for each $i\in[k]$, be such that any MBST on points $p_1,p_2,\ldots, p_n$ 
	has a bottleneck edge of length $\delta^*$. Let $T'$ be a MBST of the complete graph $K_{\{p_1,\ldots, p_n\}}$
	such that each vertex of $T'$ has degree at most 5. Such a tree $T'$ exists by Claim~\ref{degree}. 
	Because of Lemma~\ref{components} and because $\delta\geq\delta^*$, 
	there is no edge in $T'$ between points from distinct components from $\CC$. 
	Let $\Gamma$ be a graph with $k+\ell$ nodes $\CC\cup\SSS$ and the following edges:
	\begin{itemize}
		\item $s_i s_j$, whenever $p_i p_j$ is an edge in $T'$,
		\item $s_i C_j$, whenever there exists a point $p$ in $C_j$ such that $p_i p$ is an edge in $T'$.
	\end{itemize}
	Such a graph $\Gamma$ is not necessarily a tree, however we claim that it is connected. 
	This follows from the fact that any path in $T'$ can be mapped into a walk in $\Gamma$ 
	by replacing each vertex from a component $C_i$ with the component $C_i$, each vertex from a line segment $s_j$ 
	with the line segment $s_j$, and then deleting the potential consecutive repetitions of vertices of $\Gamma$. 
	Because the tree $T'$ has maximum degree 5, it follows by definition of edges of $\Gamma$ 
	that the degree of each node $s_i\in \SSS$ of $\Gamma$ has degree at most 5. 
	Let $\tau$ be any spanning tree of $\Gamma$. It is clear that $\tau$ is a $\delta$-realizable topology tree.
\end{proof}

Recall that Lemma~\ref{ell_bound} states that, in case $\delta\ge \delta^*$, it holds that $\ell\leq 4k+1$.

\begin{lemma}
\label{number_of_topologies_k}
    If $\ell\leq 4k+1$, then there are at most $\left(\OO(k)\right)^{5k}$ topology trees
    and they can be generated in $\left(\OO(k)\right)^{5k}$ time.
\end{lemma}
\begin{proof}
	Using Cayley's formula for the number of spanning trees of a labeled complete graph, 
	we get that that the number of topology trees is at most 
	$(k+\ell)^{k+\ell-2}\le \left(5k+1\right)^{5k-1}$. 
	To construct a topology tree, it is enough to determine the neighbors of each node from $\SSS$. 
	Because each node from $\SSS$ has degree at most 5 in a topology tree, 
	we can use brute force to generate all topology trees on $k+\ell\le 5k+1$ 
	nodes in $\left(\OO(k)\right)^{5k}$ time.
\end{proof}

Note that, even with more careful estimates, we could not get a bound below $k^{\OO(k)}$
because there are at least $\frac{k!}{2}=k^{\Omega(k)}$ ways to construct a path of length $k$ 
out of elements of $\SSS$.

\refstepcounter{steps}
\paragraph{DStep~\thesteps.}\label{step:inside_loop} 
Given a topology tree $\tau$, in this step we will verify whether it is $\delta$-realizable.
Therefore, because of Lemma~\ref{number_of_topologies_k}, we will execute this step $k^{\OO(k)}$ times.
For the discussion, we consider a fixed topology tree $\tau$.

We observe that if a node $C$ from $\CC$ is a separating vertex of $\tau$, which is equivalent to saying that $C$ has degree at least 2 in $\tau$, then we can treat each part of the tree $\tau$ that is ``separated'' by $C$ independently. This motivates the following definition of a \emph{significant topology subtree} of $\tau$. 

If we remove the nodes $\CC$ from $\tau$, we get a forest. To each tree $\tau'$ in this forest, 
we add all nodes from $\CC$ that are adjacent in $\tau$ to some vertex in $\tau'$, with the corresponding edges. 
The resulting tree is a \emph{significant topology subtree} of $\tau$. 
See Figure~\ref{fig:topologytree}.
Let us state a few observations about significant topology subtrees that can be checked easily.
\begin{enumerate}[a)]
    \item Each significant topology subtree of $\tau$ is an induced subtree of $\tau$.
    \item Each node from $\CC$ that is part of a significant topology subtree $\tau'$ of $\tau$, 
		has degree exactly 1 in $\tau'$. If there is another vertex of degree 1 in $\tau'$, 
		it must have degree 1 in $\tau$ as well.
	\item Each significant topology subtree has maximum degree at most $5$.
    \item The union of all significant topology subtrees of $\tau$ is the whole $\tau$.
    \item An intersection of any two significant topology subtrees of $\tau$ is either empty 
		or a graph with one node from $\CC$.
	\item Each node from $\SSS$ belongs to exactly one significant topology subtree of $\tau$.
    \item There are at most $k$ significant topology subtrees of $\tau$.
\end{enumerate}

The definition of $\delta$-realizability can now be naturally extended to significant topology subtrees. 
We say that a significant topology subtree $\tau'$ is \emph{$\delta$-realizable}, 
if there exist points $p_j\in s_j$, for each vertex $s_j\in\SSS$ from $\tau'$, such that
\begin{enumerate}[a)]
    \item for each edge $s_i s_j$ in $\tau'$, where $s_i,s_j\in \SSS$, it holds $d(p_i,p_j)\leq\delta$, and
    \item for each edge $C_i s_j$ in $\tau'$, where $C_i\in \CC$ and $s_j\in \SSS$, there exists a point $p\in C_i$ such that $d(p,p_j)\leq\delta$.
\end{enumerate}

\begin{lemma}
\label{significant}
	The topology tree $\tau$ is $\delta$-realizable if and only if 
	all of its significant topology subtrees are $\delta$-realizable.
\end{lemma}
\begin{proof}
	If the topology tree $\tau$ is $\delta$-realizable, it is clear that 
	each of the significant topology subtrees of $\tau$ is also $\delta$-realizable.
	To prove the opposite direction, assume that all of significant topology subtrees 
	of $\tau$ are $\delta$-realizable. For each significant topology subtree $\tau'$ of $\tau$, 
	we can choose points $p_j\in s_j$ in each node $s_j\in \SSS$ of $\tau'$, such that
	\begin{enumerate}[\indent a)]
		\item for each edge $s_i s_j$ in $\tau'$, where $s_i,s_j\in \SSS$, 
			it holds $d(p_i,p_j)\leq\delta$, and
		\item for each edge $C_i s_j$ in $\tau'$, where $C_i\in \CC$ and $s_j\in \SSS$, 
			there exists a point $p\in C_i$ such that $d(p,p_j)\leq\delta$.
	\end{enumerate}
	This way we uniquely defined points $p_j\in s_j$, for each $j\in[k]$. This is because, 
	for each $s_j\in\SSS$, there exists exactly one significant topology subtree $\tau'$ of $\tau$ 
	that has $s_j$ as node. It is clear that such a choice of points $p_j\in s_j$, over all $j\in [k]$, 
	shows that the topology tree $\tau$ is $\delta$-realizable.
\end{proof}

We just showed that, to describe \textbf{DStep~\ref{step:inside_loop}}, 
it is enough to describe how to verify whether a given significant topology subtree $\tau'$ of $\tau$ 
is $\delta$-realizable. To describe the latter, we restrict our attention to a fixed 
significant topology subtree $\tau'$.
Let the set of nodes of $\tau'$ be $V'\subseteq \SSS\cup\CC$. 
We denote $\SSS'=\SSS\cap V'$ and $\CC'= \CC\cap V'$. Therefore $V'$ is the disjoint
union of $\SSS'$ and $\CC'$. By definition of $\tau'$, we know that $\SSS'$ is not empty. 
Let us choose a root $s_r\in\SSS'$ for $\tau'$. 
For each segment $s_i\in\SSS'$, let $\tau'(s_i)$ be the subtree
of $\tau'$ rooted at $s_i$. In particular $\tau'(s_r)=\tau'$.

Next, we will use dynamic programming bottom-up along $\tau'$ to compute
the possible locations of points $p_i\in s_i$ on line segments $s_i\in\SSS'$ 
that can yield $\delta$-realizability for the subtree of $\tau'$ rooted at $s_i$. 
More exactly, for each $s_i\in \SSS'$, we define
\begin{align*}
	X_i = \{ p_i\in s_i\mid~ &\text{we can select one point $q_\ell\in C_\ell$, 
										for each node $C_\ell\in \CC\cap V(\tau'(s_i))$,}\\
							&\text{and one point $p_j\in s_j$, 
										for each node $s_j\in \SSS\cap V(\tau'(s_i))$ with $j\neq i$,}\\
							&\text{such that for each edge $s_\ell s_j$ of $\tau'(s_i)$ 
										we have $d(p_\ell,p_j)\le \delta$}\\
							&\text{and for each edge $s_j C_\ell$ of $\tau'(s_i)$ 
										we have $d(p_j,q_\ell)\le \delta$}\}.
\end{align*}

We begin with leaves of $\tau'$. 
If we have a leaf $s_i$ from $\SSS'$, we then have $X_i=s_i$. 
For internal nodes of $\tau'$, we have the following recursive property.

\begin{lemma}
\label{le:internal}
	Let $s_i$ be an internal node in $\tau'$ from $\SSS'$. 
	Reindexing the nodes, if needed, let us assume that the children of $s_i$ 
	in $\tau'$ are	$s_1,\ldots,s_t$ and $C_1,\ldots, C_u$. 
	Then 
	\[ 
		X_i=\bigcap_{\ell=1}^{t} \big(X_{\ell}\oplus D(0,\delta)\big)\bigcap_{\ell=1}^{u} \big(C_\ell\oplus D(0,\delta)\big)
			\bigcap s_i.
	\]
\end{lemma}
\begin{proof}
	Consider any point $p_i$ in $X_i$. From the definition of $X_i$ this means that
	we can select one point $q_\ell\in C_\ell$, for each node $C_\ell\in \CC\cap V(\tau'(s_i))$
	and one point $p_j\in s_j$ for each node $s_j\in \SSS\cap V(\tau'(s_i))$ with $j\neq i$ such that
	\begin{itemize}
		\item for each edge $s_j s_\ell$ of $\tau'(s_i)$ we have $d(p_j,p_\ell)\le \delta$, and
		\item for each edge $s_j C_\ell$ of $\tau'(s_i)$ we have $d(p_j,q_\ell)\le \delta$.
	\end{itemize}
	Looking at the edges connecting $s_i$ to its children, we obtain 
	\begin{itemize}
		\item for each $\ell\in [t]$ we have $d(p_i,p_\ell)\le \delta$, and
		\item for each $\ell\in [u]$ we have $d(p_i,q_\ell)\le \delta$.
	\end{itemize}
	Moreover, because the definition of $X_i$ includes a condition for the whole subtree $\tau'(s_i)$
	and, for each child $s_\ell$ of $s_i$, we have $\tau'(s_\ell)\subset \tau'(s_i)$,
	we have 
	\begin{itemize}
		\item for each $\ell\in [t]$, the point $p_\ell$ belongs to $X_\ell$.
	\end{itemize}
	We conclude that the point $p_i$ belongs to $\bigcap_{j=1}^{t} \big(X_{j}\oplus D(0,\delta)\big)$.
	Because for each $\ell\in [u]$ we also have $q_\ell\in C_\ell$, we also 
	conclude that $p_i$ belongs to $\bigcap_{\ell=1}^{u} \big(C_\ell\oplus D(0,\delta)\big)$.
	This finishes the proof that $X_i$ is included in the right hand side.
	
	To see the other inclusion, consider one point $p_i$ on the right hand side of the equality we want to prove.
	We then have:
	\begin{itemize}
		\item $p_i\in s_i$;
		\item for each $\ell\in [t]$, the point $p_i$ belongs to $X_{\ell}\oplus D(0,\delta)$;
		\item for each $\ell\in [u]$, the point $p_i$ belongs to $C_\ell\oplus D(0,\delta)$.
	\end{itemize}
	We can rewrite these properties as
	\begin{itemize}
		\item $p_i\in s_i$;
		\item for each $\ell\in [t]$, there is some point $p_\ell \in X_\ell$ such that $d(p_i,p_\ell)\le \delta$;
		\item for each $\ell\in [u]$, there is some point $q_\ell \in C_\ell$ such that $d(p_i,q_\ell)\le \delta$.
	\end{itemize}
	For each $\ell\in [t]$, the property that $p_\ell\in X_\ell$ implies that we can find points in
	all the nodes in $\tau'(s_\ell)$ satisfying the definition for $X_\ell$.
	Since the subtrees $\tau'(s_1), \ldots, \tau'(s_t)$ are disjoint, the selection of points for those
	subtrees are for different nodes, and thus they do not interact. The points $p_i,p_1,\ldots,p_t, q_1,\dots,q_u$
	and the ones selected for the condition of $X_1, \ldots, X_t$ certify that $p_i\in X_i$ because
	each edge of $\tau'(s_i)$ appears in one of the conditions.
\end{proof}

All geometrical computations needed to compute $X_i$ are described in Section~\ref{sec:geometry}. 
We can compute $X_i$ with $t$ operations $\big(X_\ell\oplus D(0,\delta)\big)\cap s_i$, 
described in~\ref{geom:segmentationsegment}), 
$u$ operations $\big(C_\ell\oplus D(0,\delta)\big)\cap s_i$ described in~\ref{geom:unionsegment}),
and $t+u-1$ intersections of segmentations described in~\ref{geom:segmentations}). 
For each $\ell\in [u]$, the size of the segmentation $\big(C_\ell\oplus D(0,\delta)\big)\cap s_i$ 
is at most $|C_\ell|$, and using induction on the structure of $\tau'$ (the base of the induction are the leaves), 
we can see that each point from $\PP$ contributes at most one segment to $X_i$. 
Using that $\tau'$ has at most $k$ leaves from $\SSS$, we see that the size of $X_i$ is at most 
$|\PP|+|\SSS|=n$. 
Finally, note that $t+u \le 5$ because $\tau'$ has maximum degree at most $5$.
Assuming that $X_\ell$ is already available for each child $s_\ell\in \SSS'$ of $s_i$ (thus for all $\ell\in [t]$), 
and assuming that the Voronoi diagrams on $s_i$ for $C_\ell$ are available for each $\ell\in [u]$,
we can compute $X_i$ in $\OO(n)$ time. Recall that the Voronoi diagrams inside $s_i$ for each $C_\ell$ were 
computed in \textbf{DStep~\ref{step:Voronoi}}, and thus are available.
 
The significant topology subtree $\tau'$ is $\delta$-realizable if and only if $X_r$ is not empty.
We can compute the sets $X_i$ for all $s_i\in \SSS'$ bottom-up.
At each node of $\tau'$ from $\SSS'$ we spend $\OO(n)$ time.

We have to repeat the test for each significant topology subtree.
Recall that, by Lemma~\ref{significant}, a topology tree is $\delta$-realizable
if and only if all its significant topology subtrees are $\delta$-realizable.
Since each node of $\SSS$ appears exactly in one significant topology subtree, 
for each $s_j\in \SSS$ we compute the corresponding set $X_j$ exactly once.
It follows that we spend $\OO(kn)$ time for a topology tree.
We summarize.

\begin{lemma}
\label{le:onetree}
	Assume we have already performed \textbf{DStep~\ref{step:build_tree}} 
	and \textbf{DStep~\ref{step:Voronoi}}. 
	For any given topology tree $\tau$, we can decide whether $\tau$ is $\delta$-realizable performing
	$\OO(kn)$ operations. Here, an operation may include manipulating a number 
	that has a computation tree of depth $\OO(k)$ whose internal nodes are additions, 
	subtractions, multiplications, divisions or square root computations and whose
	leaves contain input numbers (including $\delta$).
\end{lemma}
\begin{proof}
	Correctness and the bound on the number of operations follows from the discussion.
	It only remains to show the property about the computation tree of numbers.
	Each component of a segmentation $X_i$ that corresponds to a node $v$ in some rooted topological subtree $\tau'$ of $\tau$ is computed with $\OO(1)$ operations from input numbers or components of the segmentations that correspond to the children of $v$ in $\tau'$. Since each significant topology tree $\tau'$ has depth
	at most $k+1$, the claim follows.
	(Note that the numbers may participate in many more comparisons.)
\end{proof}

In \textbf{DLoop}, we try each topology tree $\tau$ 
and perform \textbf{DStep~\ref{step:inside_loop}} for $\tau$. 
If for some topology tree we find that it is $\delta$-realizable, we return \textsl{TRUE}.
If the loop finishes without finding any $\delta$-realizable topology tree, 
we return \textsl{FALSE}.

Because of our future use in the optimization version of the problem, 
we decouple the running time of \textbf{DStep~\ref{step:build_tree}} 
and \textbf{DStep~\ref{step:Voronoi}}.

\begin{theorem}
\label{decision_theorem}
	Assume that we have an instance for \connectivity with $k$ line segments and $n-k$ points ($\delta$ is not part of the input).
	After a preprocessing of $\OO(k^2 n \log n)$ time, for any given $\delta$,
	we can solve the decision version \dconnected performing $k^{\OO(k)}n $ operations. 
	Here, an operation may include manipulating a number that has a computation tree of depth $\OO(k)$ whose
	internal nodes are additions, subtractions, multiplications, divisions or square root computations
	and whose leaves contain input numbers (including $\delta$).
\end{theorem}
\begin{proof}
	We have shown that \textbf{DStep~\ref{step:inside_loop}} requires linear number of steps in $n$ and that the number of repetitions of \textbf{DLoop} depends only on $k$. Hence, to prove the theorem, we need to reduce the $\log n$ factor from \textbf{DStep~\ref{step:build_tree}} and \textbf{DStep~\ref{step:Voronoi}} with preprocessing.

    The main part of \textbf{DStep~\ref{step:build_tree}} is computing a MST on $n-k$ points, which takes $\OO(n\log n)$ time and is independent of $\delta$. Hence, it can be done with preprocessing. The rest of \textbf{DStep~\ref{step:build_tree}} (ie. defining the clusters $\CC$) can be implemented in time $\OO(k)$ for any $\delta$.

	For \textbf{DStep~\ref{step:Voronoi}}, we observe that $\delta$ can be classified into $\OO(k)$ different
	intervals of values that will give the same clusters $\CC$ and for which \textbf{DStep~\ref{step:Voronoi}}
	is the same. More precisely, let $e_1,\ldots, e_{4k+1}$ be 
	$4k+1$ longest edges in the MST $T$ for $\PP$, obtained after preprocessing for \textbf{DStep~\ref{step:build_tree}} described above, sorted such that $|e_1|\ge |e_2|\ge \ldots\ge |e_{4k+1}|$.
	For each $i\in [4k]$ and each $\delta$ in the interval $(|e_i|,|e_{i+1}|]$ we will have the same family
	$\CC$ of $i+1$ connected componenents, namely those in the graph $T-\{ e_1,\dots, e_i \}$.
	For each $\delta \ge |e_1|$, we have a single component in $\CC$.
	For each $\delta< |e_{4k+1}|$, we know that $\delta < \delta^*$ because of Lemma~\ref{ell_bound}.
	Therefore, we can consider the $\OO(k)$ different connected components that 
	appear in the graphs $T_0=T$ and $T_i=T_{i-1}-e_i$, where $i\in [4k]$.
	For each such connected component $C$ and each segment $s\in \SSS$,
	we compute the Voronoi diagram on $s$ of the points of $C$ using \textbf{DStep~\ref{step:Voronoi}}.
	In total we compute $\OO(k^2)$ Voronoi diagrams, and each of them takes $\OO(n \log n)$ time. This can all be done with preprocessing, hence \textbf{DStep~\ref{step:Voronoi}} can be implemented in time $\OO(1)$ for any $\delta$.
	
    The depths of the computation trees of numbers used in \textbf{DStep~\ref{step:build_tree}} 
	and \textbf{DStep~\ref{step:Voronoi}} are $\OO(1)$. 

	Consider now that we are given a value $\delta$ after the just described preprocessing. If $\delta< |e_{4k+1}|$, we return \textsl{FALSE}.	Otherwise, \textbf{DStep~\ref{step:build_tree}} and \textbf{DStep~\ref{step:Voronoi}} now require only $\OO(k)$ time. We perform \textbf{DLoop} iterating over all topology trees.
    The correctness is proven with Lemma~\ref{topology} and Lemma~\ref{le:onetree}.
	
	By Lemma~\ref{number_of_topologies_k}, \textbf{DStep~\ref{step:inside_loop}} is repeated $k^{\OO(k)}$ times,
	and each such iteration performs $\OO(k n)$ operations because of Lemma~\ref{le:onetree}.
	In total we perform $k^{\OO(k)}n$ operations. 
	Numbers in each iteration of \textbf{DLoop} are computed independently of the numbers computed in another iteration,
	and therefore we can use the bound on the depth of computation trees of Lemma~\ref{le:onetree} for each of them.
\end{proof}

\begin{corollary}
\label{co:decision}
	The decision problem \dconnected for $k$ line segments and $n-k$ points can be solved
	performing $k^{\OO(k)}n \log n $ operations. 
	Here, an operation may include a number that has a computation tree of depth $\OO(k)$ whose
	internal nodes that are additions, subtractions, multiplications, divisions or square root computations
	and whose leaves contain input numbers (including the input value $\delta$).
\end{corollary}

\section{Introducing $h$-square root functions}
\label{sec:root_functions}

When using parametric search, we will need to trace the boundary of the segmentations $X_i(\delta)$ 
as a function of $\delta$. In this section we introduce and discuss the properties
of the functions that will appear.

For any natural\footnote{We define that $0$ is a natural number.} number $h$, we define $h$-square root functions recursively. A 0-square root function is any linear function. For $h\geq 1$, an $h$-square root function is any function of the form $f(x)=a_1g(x)+a_2+a_3\sqrt{\pm x^2+a_4g(x)^2+a_5g(x)+a_6}$, where $a_1,a_2,a_3,a_4,a_5,a_6\in\RR$,  
and $g(x)$ is a $(h-1)$-square root function.  
The domain of an $h$-square root function is all such $x\in\RR$ for which all the square roots that appear inside them have non-negative arguments. Note that an $h$-square root function is also an $h'$-square root function
for all $h'\ge h$ because we may take $a_1=1$ and $a_2=a_3=0$.

The following lemma presents the setting where we will meet the $h$-square root functions. Note that this setting can occur in computations in~\ref{geom:circleline}), that is, when computing the intersection of a circle with a line.

\begin{lemma}
\label{le:sqrt_alternative}
	Let $q,e$ and $f$ be vectors in $\RR^2$, $||e||=||f||=1$ and let $g(x)$ be an $(h-1)$-square root function, 
	for some $h\in\ZZ^+$. 
	Then any continuous function $t(x)$ that solves the equation
	$$||q+t(x)e-g(x)f||^2=x^2$$
	is an $h$-square root function.
\end{lemma}
\begin{proof}
	The equation
	$||q+t(x)e-g(x)f||^2=x^2$ is equivalent to 
	\[
		t(x)^2+\bigl[ 2(q-g(x)f)\cdot e\bigr] t(x)+||q-g(x)f||^2-x^2=0.
	\]
	The discriminant of this quadratic equation in $t(x)$ is 
	\begin{align*}
	\Delta(x)&=4((q-g(x)f)\cdot e)^2-4(||q-g(x)f||^2-x^2)\\
	 & = 4\Bigl((q\cdot e)^2-2(q\cdot e)(f\cdot e)g(x)+(e\cdot f)^2 g(x)^2-||q||^2+2(q\cdot f)g(x)-g(x)^2+x^2\Bigr)\\
	 & = 4\Bigl( x^2+ \bigl[(e\cdot f)^2 -1 \bigr] g(x)^2 + \bigl[2(q\cdot e)(f\cdot e) + 2(q\cdot f)  ] g(x) + \bigl[ (q\cdot e)^2-||q||^2\bigr]\Bigr).
	\end{align*}
	If we denote $\Tilde{\Delta}(x)=\frac{1}{4}\Delta(x)$, we have, 
	for all $x$ in the domain of $g$ for which $\Delta(x)\geq 0$,
	\begin{align*}
			t_1(x)&=e\cdot(g(x)f-q)-\sqrt{\Tilde{\Delta}(x)} = \bigl[e\cdot f\bigr] g(x) - \bigl[e\cdot q\bigr]-\sqrt{\Tilde{\Delta}(x)},\\
			t_2(x)&=e\cdot(g(x)f-q)+\sqrt{\Tilde{\Delta}(x)} = \bigl[e\cdot f\bigr] g(x) - \bigl[e\cdot q\bigr]+ \sqrt{\Tilde{\Delta}(x)},
		\end{align*}
	which are both $h$-square root functions.
\end{proof}

The next lemma will help us solve equations with $h$-square root functions.
\begin{lemma}
\label{le:sqrt_solutions}
	Let $f(x)$ and $g(x)$ be $h$-square root functions for $h\in\NN$. 
	Then all of the solutions of $f(x)=g(x)$ are also roots of a polynomial of degree at most $4^h$ in $x$. 
	The coefficients of this polynomial can be computed from parameters in $f$ and $g$ in $2^{\OO(h)}$ 
	steps by using only multiplications, additions and subtractions.
\end{lemma}
\begin{proof}
    Let $f_h(x)=f(x)$ be an $h$-square root function obtained from an $(h-1)$-square root function $f_{h-1}(x)$, 
	which was obtained from an $(h-2)$-square root function $f_{h-2}(x)$, \ldots , 
	which was obtained from a $0$-square root function $f_0(x)$. 
	In a similar way we define the functions $g_h(x),g_{h-1}(x),\ldots,g_0(x)$.

    We will transform the equation $f_h(x)=g_h(x)$ into the desired polynomial equation 
	by squaring it at most $2h$ times, each time also rearranging the terms a bit and using
	the replacement rule $\sqrt{u}^2=u$ multiple times. These transformations may introduce additional solutions, 
	but we keep all the original solutions.

	We show by induction on $i$ that, for each $i=0,\dots, h$, 
	there is a polynomial $P_i(X,Y_i,Z_i)$ of degree $4^i$ such
	that the solutions of $P_i(x,f_{h-i}(x),g_{h-i}(x))=0$ include the solutions of $f(x)=g(x)$.
	For the base case, $i=0$, it is obvious that the polynomial $P_0(X,Y_i,Z_i)=Z_i-Y_i$
	satisfies the condition because $P_0(x,f_h(x),g_h(x))=0$ is equivalent to $f_h(x)-g_h(x)=0$.
	
	Assume that we have the polynomial $P_i(X,Y_i,Z_i)$ for some $i$. We show how to compute 
	$P_{i+1}(X,Y_{i+1},Z_{i+1})$. The polynomial $P_i$ can be written as 
	\[
		P_i(X,Y_i,Z_i) = \sum_{%
						\begin{minipage}{2.5cm}\centering\tiny
							$\alpha+\beta+\gamma\le 4^i$\\
							$\alpha, \beta,\gamma\in \NN$
						\end{minipage}} c_{\alpha,\beta,\gamma}\,X^\alpha Y_i^\beta Z_i^\gamma,
	\]
	for some coefficients $c_{\alpha,\beta,\gamma}\in \RR$,
	and we have $P_i(x,f_{h-i}(x),g_{h-i}(x))=0$.
	We substitute in the latter equation $f_{h-i}(x)$ and $g_{h-i}(x)$ by their definition
	using $f_{h-i-1}(x)$ and $g_{h-i-1}(x)$, respectively.
	More precisely, and to shorten the expressions, we have for some $a_1,\ldots,a_6,b_1,\ldots,b_6\in \RR$,
	\begin{align*}
		A_1(x) &= a_1 f_{h-i-1}(x)+a_2\\
		A_2(x) &= \pm x^2 + a_4 f_{h-i-1}(x)^2 + a_5 f_{h-i-1}(x)+a_6\\
		f_{h-i}(x) &= A_1(x) + a_3 \sqrt{A_2(x)}\\
		B_1(x) &= b_1 g_{h-i-1}(x)+b_2 \\
		B_2(x) &= x^2 + b_4 g_{h-i-1}(x)^2 + b_5 g_{h-i-1}(x)+b_6\\
		g_{h-i}(x) &= B_1(x) + b_3 \sqrt{B_2(x)}.
	\end{align*}
	We thus get the equation
	\begin{align*}
			0 ~=~ \sum_{%
						\begin{minipage}{2.5cm}\centering\tiny
							$\alpha+\beta+\gamma\le 4^i$\\
							$\alpha, \beta,\gamma\in \NN$
						\end{minipage}} c_{\alpha,\beta,\gamma}\,x^\alpha \Bigl(A_1(x) + a_3 \sqrt{A_2(x)}\Bigr)^\beta \Bigl(B_1(x) + b_3 \sqrt{B_2(x)}\Bigr)^\gamma .
	\end{align*}
	We expand the terms $(A_1(x) + a_3 \sqrt{A_2(x)})^\beta$ and $(B_1(x) + b_3 \sqrt{B_2(x)})^\gamma$ 
	using the binomial theorem and, in the resulting equation, we replace 
	\begin{itemize}
		\item each term $(\sqrt{A_2(x)})^{\beta'}$ with even $\beta'$ by $A_2(x)^{\beta'/2}$;
		\item each term $(\sqrt{A_2(x)})^{\beta'}$ with odd $\beta'$ by $A_2(x)^{(\beta'-1)/2} \sqrt{A_2(x)}$;
		\item each term $(\sqrt{B_2(x)})^{\gamma'}$ with even $\gamma'$ by $A_2(x)^{\gamma'/2}$;
		\item each term $(\sqrt{B_2(x)})^{\gamma'}$ with odd $\gamma'$ by $A_2(x)^{(\gamma'-1)/2} \sqrt{B_2(x)}$.
	\end{itemize}
	We group the terms with a factor $\sqrt{A_2(x)}$ or $\sqrt{B_2(x)}$ remaining.
	For some polynomial $Q_1(X, Y, Z)$ of degree at most $4^i$, polynomials $Q_2(X, Y, Z)$ and $Q_3(X, Y, Z)$  of degree at most $4^{i}-1$ and polynomial $Q_4(X, Y, Z)$ of degree at most $4^{i}-2$ we get an equation
	\begin{align*}
			0 ~=~ &Q_1(x, f_{h-i-1}(x),g_{h-i-1}(x))+\\
				&Q_2(x, f_{h-i-1}(x),g_{h-i-1}(x)) (\sqrt{A_2(x)}) + \\
				&Q_3(x, f_{h-i-1}(x),g_{h-i-1}(x)) (\sqrt{B_2(x)}) + \\
				&Q_4(x, f_{h-i-1}(x),g_{h-i-1}(x)) (\sqrt{A_2(x)} \sqrt{B_2(x)}).
	\end{align*}
	We pass the terms with $\sqrt{A_2(x)}$ to one side and all the other terms to the other side. 
	We square the equation, expand each side, 
	replace each $(\sqrt{A_2(x)})^2$ with $A_2(x)$, and replace each $(\sqrt{B_2(x)})^2$ with $B_2(x)$.
	We are left with an equation where some terms include the factor $\sqrt{B_2(x)}$; all the other terms are polynomial in $x$,	$f_{h-i-1}(x)$ and $g_{h-i-1}(x)$.
	We collect on one side the terms with $\sqrt{B_2(x)}$, square both sides of the equation,
	and replace each $(\sqrt{B_2(x)})^2$ with $B_2(x)$.
	We are left with an equation that is polynomial in $x$, $f_{h-i-1}(x)$ and $g_{h-i-1}(x)$;
	this equation defines the polynomial $P_{i+1}(X,Y_{i+1},Z_{i+1})$.
	Since we have done reorganizations and have squared both sides of the equation twice,
	the degree of the polynomial $P_{i+1}$ is at most $4$ times the degree of $P_i$.
	Therefore $P_{i+1}$ has degree at most $4^{i+1}$.
	
	For $i=h$, we obtain a polynomial $P_{h}(X,Y_h,Z_h)$ of degree at most $4^h$ such that 
	the solutions to $P_{h}(x,f_0(x),g_0(x))=0$ contains the solutions for $f(x)=g(x)$.
	Note that the equation $P_{h}(x,f_0(x),g_0(x))=0$ may contain some additional 
	solutions that are added through the algebraic manipulation,
	possibly also solutions that are not in the domains of $f(x)$ or $g(x)$.
	The equation $P_{h}(x,f_0(x),g_0(x))=0$ is a polynomial of degree at most  $4^h$ in $x$ because
	$f_0(x)$ and $g_0(x)$ are linear. 
	
	For each $i\in [h]$, because the polynomial $P_i(X,Y_i,Z_i)$ has degree at most $4^i$, it is defined by $2^{\OO(i)}$ coefficients,
	and each of its coefficients comes from making $2^{\OO(i)}$ operations through the computation. 
	We conclude that all the polynomials can be computed in $2^{\OO(h)}$ time.
\end{proof}

\section{Parametric version}
\label{sec:param}

In this section we will solve the initial optimisation problem \connectivity for uncertainty regions given as line segments. 
Given a set $\SSS=\{s_1,s_2,\dots ,s_k\}$ of segments and a set $\PP=\{p_{k+1},p_{k+2},\dots, p_n\}$
of points in the plane, find $\delta^*$, which is the smallest $\delta\geq 0$, such that the decision problem \dconnected on inputs $\SSS$, $\PP$ and $\delta$ has the answer \textsl{TRUE}. 
We can shortly write $\delta^*=\connectivity(\SSS,\PP)=\min\{\delta \mid \dconnected(\SSS,\PP,\delta)\}.$

We will use parametric search. The idea is to simulate the decision algorithm described in Section~\ref{sec:decision}
for the unknown value $\delta^*$. Through the algorithm we maintain two values $\delta_m<\delta_M$
such that the interval $(\delta_m,\delta_M]$ contains $\delta^*$ and such that, for any $\delta\in (\delta_m,\delta_M)$, 
the algorithm branches in the same way, that is, the combinatorial decisions of the algorithm
are the same.
Thus, the algorithm has the same outline as it was used for describing the algorithm for the problem \dconnected. This means that it will be given in 4 parts: \textbf{Step~\ref{step:build_tree}}, \textbf{Step~\ref{step:Voronoi}}, \textbf{Loop} and \textbf{Step~\ref{step:inside_loop}}. These 4 parts will be analogous to the parts with the corresponding names in the algorithm for the problem \dconnected.

Throughout our algorithm, we will constantly update $\delta_m$ and $\delta_M$ such that the value of $\delta_m$ will never decrease, the value of $\delta_M$ will never increase and $\delta^*$ will be in the interval $(\delta_m,\delta_M]$. 
We will mostly update $\delta_m$ and $\delta_M$ by using \emph{parametric search among some set of values} $\Delta=\{\delta_1,\delta_2,\ldots, \delta_N\}$.  This means that we will discard the values from $\Delta$ outside the interval $(\delta_m,\delta_M)$ and, for the sake of simpler description, we will add the values $\delta_m$ and $\delta_M$ to $\Delta$. 
Then, we will sort the values in $\Delta$ and we will do a binary search to determine two consecutive values $\delta'_1<\delta'_2$, such that  $\dconnected(\SSS,\PP,\delta'_1)=\textsl{FALSE}$ and $\dconnected(\SSS,\PP,\delta'_2)=\textsl{TRUE}$. 
We will update the values $\delta_M=\delta'_2$ and $\delta_m=\delta'_1$. Clearly, it will hold $\delta_m<\delta^*\leq\delta_M$ and none of the values that were initially in $\Delta$ will be in the interval $(\delta_m,\delta_M)$. 
For this step in parametric search, we spend $\OO(N\log N)$ time plus the time needed to solve $\OO(\log N)$ decision problems.

We use the preprocessing of Theorem~\ref{decision_theorem}: after a preprocessing of $\OO( k^2 n \log n)$ time,
we can solve each decision problem performing $k^{\OO(k)}n$ operations. 
The preprocessing is performed only once. Afterwards, each parametric search among a set of $N$ values takes $\OO(N\log N)+k^{\OO(k)}n \log N$ steps.

\paragraph{Step~\ref{step:build_tree}.} To get an upper bound on $\delta^*$, we first choose arbitrary points 
$p_1\in s_1,\ldots, p_k\in s_k$ and compute a MST on points $p_1,p_2,\ldots, p_n$. 
We define $\delta_M$ as the maximum length of an edge in this MST. 
To set a proper lower bound on $\delta^*$, we run the decision algorithm for $\delta=0$ and, 
if it returns \textsl{TRUE}, we return $\delta^*=0$. Otherwise we define $\delta_m=0$. 
We see that $\delta_m<\delta^*\leq\delta_M$. 

To continue with \textbf{Step~\ref{step:build_tree}}, we compute a minimum spanning tree $T$ for the $n-k$ points in $\PP$. Let $e_1,\ldots, e_{n-k-1}$ be the edges of $T$ sorted by length such that $|e_1|\geq |e_2|\geq\cdots \geq |e_{n-k-1}|$. 
We do a parametric search among the values $|e_1|,|e_2|,\ldots, |e_{\min\{n-k-1,4k+1\}}|$ to update $\delta_m$ and $\delta_M$. 

Next, we remove all edges of $T$ that are at least as long as $\delta_M$. By Lemma~\ref{ell_bound}, 
if there are any remaining edges in $T$, $\delta_m$ is the length of the longest of the remaining edges. 
Let the remaining connected components of the tree $T$ be $\CC=\{C_1,C_2,\ldots, C_{\ell}\}$. 
Note that we removed exactly $\ell-1$ edges. If $\ell>4k+1$, we return $\delta^*=\delta_M$. 
This can be done because of Lemma~\ref{ell_bound}. The following lemma clearly holds.

\begin{lemma}
\label{lemma:Step1}
	For each $\delta\in (\delta_m,\delta_M)$, the algorithm for the problem $\dconnected$ on input $(\SSS,\PP,\delta)$ produces
	the same MST $T$ and the set of components $\CC$ in \textbf{DStep~\ref{step:build_tree}} as 
	were obtained after \textbf{Step~\ref{step:build_tree}}.
\end{lemma}

In \textbf{Step~\ref{step:build_tree}} we used the algorithm for the decision problem \dconnected $\OO(\log k)$ times. 
Hence, \textbf{Step~\ref{step:build_tree}} runs in time $\OO(n\log n) + k^{\OO(k)}n\log k = k^{\OO(k)}n\log n$.

\paragraph{Step~\ref{step:Voronoi}.}
As in \textbf{DStep~\ref{step:Voronoi}}, at the end of this step, we would like, for each component $C_i\in\CC$ and for each line segment $s_j\in\SSS$, to have the Voronoi diagram on $s_j$ of the points in $C_i$. Note that this was already computed during the preprocessing of Theorem~\ref{decision_theorem}, this means before  \textbf{Step~\ref{step:build_tree}}. Hence, we do nothing on this ``step''.


\paragraph{Loop.} We treat each line segment from $\SSS$ and each component from $\CC$ as an abstract vertex and 
we iterate over all topology trees $\tau$ on these $k+\ell$ vertices.

\paragraph{Step~\ref{step:inside_loop}.} We simulate \textbf{DStep~\ref{step:inside_loop}} while doing parametric search. 
Given a topology tree $\tau$, we iterate over all of its significant topological subtrees. 
We restrict our attention to one fixed significant topological subtree $\tau'$.
Let the set of vertices of $\tau'$ be $V'\subseteq \SSS\cup\CC$. 
We denote $\SSS'=\SSS\cap V'$ and $\CC'= \CC\cap V'$. 
By definition of $\tau'$, we know that $\SSS'$ is not empty. 
Let us choose a root $s_r\in\SSS'$ of $\tau'$. For each node $s_i\in\SSS'$ of $\tau'$, 
let $\tau'(s_i)$ be the subtree of $\tau'$ rooted at $s_i$,
and let its height $h(s_i)\in\NN$ be the number of edges on a longest path in $\tau'$ that begins in $s_i$ 
and is contained in $\tau'(s_i)$. Note that each such a longest path must end in a leaf of $\tau'$.

As in \textbf{DStep~\ref{step:inside_loop}}, we use dynamic programming bottom-up along $\tau'$.
For each segment node $s_i\in\SSS'$ of $\tau'$, we compute $X_i(\delta)$, as defined in \textbf{DStep~\ref{step:inside_loop}},
but taking $\delta$ as a parameter that takes values inside the interval $(\delta_m,\delta_M)$. 
It will be convenient to use that $X_i(\delta)$ increases with $\delta$: whenever $\delta'<\delta$,
we have $X_i(\delta')\subseteq X_i(\delta)$.

If we have a leaf $s_i$ from $\SSS'$, then we have $X_i(\delta)=s_i$. 
Consider now a segment node $s_i\in\SSS'$ of $\tau'$.
As in Lemma~\ref{le:internal}, we may reindex the nodes, if needed, and assume that the children of $s_i$ 
in $\tau'$ are $s_1,\ldots,s_t$ and $C_1,\ldots, C_u$. 
Then, by Lemma~\ref{le:internal} we have
\[ 
	X_i(\delta)=\bigcap_{\ell=1}^{t} \big(X_{\ell}(\delta)\oplus D(0,\delta)\big)\bigcap_{\ell=1}^{u} \big(C_\ell\oplus D(0,\delta)\big)
		\bigcap s_i.
\]
We will use parametric search to determine the size of the segmentation $X_i(\delta)$ and 
we will represent its components as at most $h(s_i)$-square root functions of $\delta$, 
as defined in Section~\ref{sec:root_functions}. 
Because we process the tree $\tau'$ bottom-up, we can assume that the sizes of segmentations 
$X_\ell(\delta)$, for $\ell \in[t]$, are already fixed in the interval $\delta\in(\delta_m,\delta_M)$ 
and that their components are at most $(h(s_i)-1)$-square root functions of $\delta$.

\paragraph{Computing an intersection of a $\delta$-neighborhood of a segmentation with a line segment.} 
Let us first describe how we can compute each of the $t$ operations 
\[ 
	Y_\ell(\delta) = \big(X_\ell(\delta)\oplus D(0,\delta)\big)\cap s_i.
\] 
We will often leave out in the notation the dependency on $\ell$ and $i$. 
We will closely follow the algorithm~\ref{geom:segmentationsegment}) from Section~\ref{sec:geometry}.

Let the line segment $s_i$ be $s_i=(p_s, e_s, a_s, b_s)$. 
Let the segmentation $X_\ell$ be $X_\ell(\delta)=(p_X, e_X, a_1(\delta), b_1(\delta), \ldots , a_N(\delta), b_N(\delta))$, 
where $a_j(\delta)$ and $b_j(\delta)$, for $j\in[N]$ are $(h(s_i)-1)$-square root functions of $\delta$. 
For $j\in [N]$, let $\sigma_j(\delta)$ be the segment $\sigma_j(\delta)=(p_X,e_X, a_j(\delta), b_j(\delta))$,
let $\gamma_j(\delta)$ be the boundary of $\sigma_j(\delta)\oplus D(0,\delta)$,
and let $\eta_j(\delta)$ be the intersection of $s_i$ with $\sigma_j(\delta)\oplus D(0,\delta)$.
Recall Figure~\ref{fig:segmentationsegment}.

We first narrow the interval defined by $\delta_m<\delta_M$ in such a way that, for each single $j\in [N]$,
the intersection $\eta_j(\delta)$ is empty for all $\delta$ in the interval $(\delta_m, \delta_M)$ or nonempty for all $\delta$ in the interval $(\delta_m, \delta_M)$.
For each $j\in [N]$, we compute the value of $\delta_j$ such that $\eta_j(\delta_j)$ is non-empty
but $\eta_j(\delta)$ is empty for all $\delta < \delta_j$.
Because $X_\ell(\delta)$ increases with $\delta$, there is at most one single $\delta_j$ that may satisfy this condition.
It may be that $\delta_j$ does not exist because $\eta_j(\delta)$ is always non-empty; in this case
we set $\delta_j=\delta_m$.
Each such value $\delta_j$ is a solution to some equation involving the segment $s_i$ and a circle of radius $\delta$
centered at $a_j(\delta)$ or $b_j(\delta)$, or lines parallel to $e_X$ at distance $\delta$ from $\sigma_j$.
Because of Lemmas~\ref{le:sqrt_alternative} and~\ref{le:sqrt_solutions}, the value $\delta_j$ is a root
of a polynomial in $\delta$ of degree at most $4^{h(s_i)}\le 4^k$.
We then do a parametric search among the values $\{ \delta_1,\dots,\delta_N\}$ to update $\delta_m$ and $\delta_M$.
We can then assume that, for each $j\in [N]$, the segment $\eta_j(\delta)$ is empty for all $\delta$
with $\delta_m < \delta < \delta_M$ or non-empty for all $\delta$
with $\delta_m < \delta < \delta_M$.

For each $j\in [N]$, we have the (possibly empty) segment $\eta_j(\delta)=(p_s,e_s,a'_j(\delta),b'_j(\delta))$. All these segments have the same reference point $p_s$ and vector $e_s$. The functions $a'_j(\delta)$ and $b'_j(\delta)$ are at most $h(s_i)$-square root functions because of Lemma~\ref{le:sqrt_alternative}. To merge
the non-empty segments that are overlapping, we have to sort the values $a'_j(\delta), b'_j(\delta),a'_{j+1}(\delta), b'_{j+1}(\delta)$, for each single $j\in [N-1]$. For this,
we perform a step of parametric search among the solutions of the equations $a'_j(\delta) = a'_{j+1}(\delta)$, $a'_j(\delta) = b'_{j+1}(\delta)$, $b'_j(\delta) = a'_{j+1}(\delta)$ and  $b'_j(\delta) = b'_{j+1}(\delta)$, for all $j\in[N-1]$.
Because of Lemma~\ref{le:sqrt_solutions}, these solutions are roots of polynomials of degree at most $4^k$.

To summarize, spending $k^{\OO(k)}N \log N$ time to manipulate segments, polynomials of degree at most $4^k$ and their roots, 
and performing $O(\log N)$ calls to the decision problem, we have an interval $(\delta_m,\delta_M)$
where $\big(X_\ell(\delta)\oplus D(0,\delta)\big)\cap s_i$ is described by the same combinatorial structure.
In particular, it is described by a segmentation 
\[
   Y_\ell(\delta) ~=~
		(p_s,e_s,\tilde a_1(\delta),\tilde b_1(\delta),\ldots ,\tilde a_{N'}(\delta),\tilde b_{N'}(\delta))
   ~~~~\text{ for all $\delta\in (\delta_m,\delta_M)$}.
\]
Note that $N'$ depends on $\delta$ and $\ell$, but it is constant for all $\delta\in (\delta_m,\delta_M)$.

We perform this procedure for each $\ell\in [t]$, where $t\leq 5$.
If, for some $\delta\in(\delta_m,\delta_M)$ and hence for all $\delta\in(\delta_m,\delta_M)$, 
we get that $\big(X_\ell(\delta)\oplus D(0,\delta)\big)\cap s_i$ is empty,
then we know that the topology tree $\tau$ under consideration is not $\delta$-realizable for any $\delta<\delta_M$
and therefore we can move on to the next topology tree $\tau$ inside \textbf{Loop}.

\paragraph{Computing an intersection of a $\delta$-neighborhood of a set of points with a line segment.} 
Next, we describe each of the $u$ operations $Z_\ell(\delta) = \big(C_\ell\oplus D(0,\delta)\big)\cap s_i$. 
We will often leave out in the notation the dependency on $\ell$ and $i$. 
Let $(a_1, J_1),\ldots, (a_{N}, J_N)$ be a Voronoi diagram on $s_i$ for points from $C_\ell$. 
Then 
\[
	Z_\ell(\delta)=\bigcup_{j=1}^{N}\big(D(a_j,\delta)\cap J_j\big).
\]
This implies that when $\delta$ goes from $\delta_m$ to $\delta_M$, the set
$Z_\ell(\delta)$ goes through $\OO(N)$ combinatorial changes. 
This is because for each $j\in[N]$, $D(a_j,\delta)\cap J_j$ goes through at most 3 combinatorial changes: when $D(a_j,\delta)\cap J_j\neq\emptyset$ for the first time,  when one endpoint of $J_j$ is included 
in $D(a_j,\delta)$ and when both endpoints of $J_j$ are included in $D(a_j,\delta)$.
See Figure~\ref{fig:Voronoi3}.

\begin{figure}[tb]
	\centering
	\includegraphics[width=\textwidth,page=9]{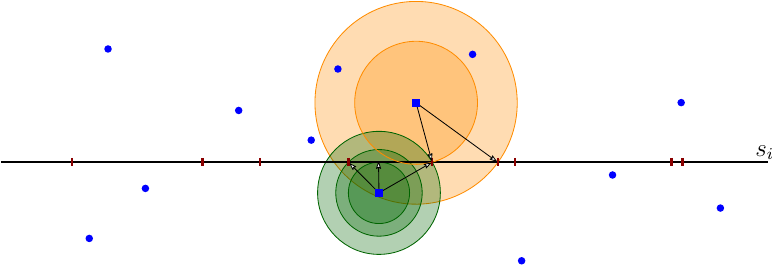}
	\caption{Example showing for two different points the moments when 
		$D(a_j,\delta)\cap J_j$ combinatorially change. In one case (orange) there are two
		changes, in the other case (green) there are three changes. The value of $\delta$
		corresponds to the radius of the disk.
		The Voronoi diagram is from Figure~\ref{fig:Voronoi}.}
	\label{fig:Voronoi3}
\end{figure}

We compute in time $\OO(N)$ all $\OO(N)$ values of $\delta$ for which these combinatorial changes occur and we do a parametric search among them to update $\delta_m$ and $\delta_M$. After that, for any $\delta\in(\delta_m,\delta_M)$, the segmentation $Z_\ell(\delta)$ has fixed size and its components are at most $1$-square root functions of $\delta$.

We perform this procedure for each $\ell\in [u]$. 
Using that $|C_\ell| \le n$ for all $\ell\in [u]$ and that $u\le 5$, we spend a total of 
$\OO(n \log n)$ time plus $\OO(\log n)$ calls to the decision problem.
If, for some $\delta\in(\delta_m,\delta_M)$ and hence for all $\delta\in(\delta_m,\delta_M)$, 
we get that for some $\ell\in [u]$ the segmentation
$Z_\ell(\delta)$ is empty,
then we know that the topology tree $\tau$ under consideration is not $\delta$-realizable for any $\delta<\delta_M$
and therefore we can move on to the next topology tree $\tau$ inside \textbf{Loop}.

\paragraph{Computing intersections of segmentations.} Let us describe how we can compute the intersections of segmentations needed to finish the computation of $X_i(\delta)$. 
As in the statement of Lemma~\ref{le:internal}, let $t$ and $u$ denote the number of children of each type for $s_i$ in $\tau'$. 
At this point we have:
\begin{itemize}
\item two values $\delta_m<\delta_M$;
\item segmentations $Y_\ell(\delta)= \big(X_\ell(\delta)\oplus D(0,\delta)\big)\cap s_i$, for all $\ell\in [t]$; and 
\item segmentations $Z_\ell(\delta)= \big(C_\ell\oplus D(0,\delta)\big)\cap s_i$, for all $\ell\in [u]$;
\end{itemize}
such that for all $\delta\in (\delta_m,\delta_M)$ each segmentation has immutable size (number of segments)
and each value describing any part of any segmentation is an $h(s_i)$-square root function.

We have to compute the intersection of these $u+t$ segmentations on $s_i$. Recall that $u+t$
is bounded by $5$ because it is the degree of $s_i$ in $\tau'$.
We do this by pairs, which means that we have to compute $t+u-1\le 4$ intersections
of pairs of segmentations. We describe how to perform the merge of two segmentations.

Consider two of the segmentations that may appear through the process:
\begin{align*}
	X(\delta) &= (p_s,e_s, a_1(\delta), b_1(\delta),\ldots , a_{N}(\delta), b_{N}(\delta)) ~~~\text{of size $N$}\\
	X'(\delta)&= (p_s,e_s, a'_1(\delta), b'_1(\delta),\ldots , a'_{N'}(\delta), b'_{N'}(\delta)) ~~~\text{of size $N'$}
\end{align*}
We want to compute $X(\delta)\cap X'(\delta)$.
For this, it suffices to sort the values 
\begin{align*}
	&a_1(\delta) \le b_1(\delta) < a_2(\delta)\le b_2(\delta)< \ldots < a_{N}(\delta)\le b_{N}(\delta), ~~\text{ and }\\
	&a'_1(\delta) \le b'_1(\delta) < a'_2(\delta)\le b'_2(\delta)< \ldots < a'_{N'}(\delta)\le b'_{N'}(\delta)
\end{align*}
for any $\delta\in (\delta_m,\delta_M)$.
After sorting the endpoints, we can easily compute the intersection $X(\delta)\cap X'(\delta)$ in $O(N+N')$ time.
Since we are merging two lists that are sorted, we can use Cole's technique~\cite{Cole87} for
parametric search on networks applied to the bitonic sorting network~\cite[Section 4.4]{JaJa-book}. 
This gives a running time of $\OO((N+N') \log (N+N'))$ plus $O(\log (N+N'))$
calls to the decision problem. Using that $N+N'=\OO(n)$, we get a running time of 
$\OO(n \log n) + k^{\OO(k)} n \log n$ for the intersection of two segmentations. 

Since Cole's technique is complex,
we provide an alternative, simpler way of achieving the same time bound to compute the intersection of two segmentations and that uses properties of our setting. The key insight is that the
segmentations we are considering do not decrease with $\delta$ in the following sense:
if $0\leq \delta_1 < \delta_2$, then $X(\delta_1)\subseteq X(\delta_2)$
and $X'(\delta_1)\subseteq X'(\delta_2)$. This follows from the definition of $X_i(\delta)$.

This implies that the functions $a_j(\delta)$, for $j\in [N]$, 
and the functions $a'_j(\delta)$, for $j\in[N']$, are (not necessarily strictly) decreasing on the interval $(\delta_m,\delta_M)$, while the functions $b_j(\delta)$, for $j\in[N]$, and the functions $b_j(\delta)$, for $j\in[N']$ are (not necessarily strictly) increasing on the interval $(\delta_m,\delta_M)$. 
Moreover, all these functions are $h(s_i)$-square root functions, defined (at least) on the interval $(\delta_m,\delta_M)$. By continuity, they are also defined on the interval $(\delta_m,\delta_M]$.

\begin{lemma}
	There are at most $4^k \OO(N+N')$ values of $\delta$ in the interval $(\delta_m,\delta_M)$
	where the boundary of some segment in $X(\delta)$ may intersect with a boundary of some segment in $X'(\delta)$. We are only considering such pairs of boundaries, where not both boundaries are constant on the interval $(\delta_m,\delta_M)$.
	These values can be computed in $2^{\OO(k)} (N+N')+ \OO((N+N') \log (N+N'))$ time.
\end{lemma}
\begin{proof}
	Let $\sigma_1(\delta),\dots, \sigma_N(\delta)$ be the segments in $X(\delta)$;
	let $\sigma'_1(\delta),\dots, \sigma'_{N'}(\delta)$ be the segments in $X'(\delta)$.
	If for some $\delta$ the boundary of some segments $\sigma_i(\delta)$ and $\sigma'_j(\delta)$ intersect,
	then, because the segments are monotonely increasing, $\sigma_i(\delta_M)$ and $\sigma'_j(\delta_M)$
	intersect.  Here we are only inserting $\delta=\delta_M$ into the boundaries of segments $\sigma_i$ and $\sigma'_j$ and we are not considering a possible combinatorial change of $X(\delta)$ or $X'(\delta)$ at $\delta=\delta_M$. This is because we are only interested in limits when $\delta\in(\delta_m,\delta_M)$ approaches $\delta_M$.
	Because $\sigma_1(\delta_M),\dots, \sigma_N(\delta_M)$ are pairwise interior disjoint,
	and $\sigma'_1(\delta_M),\dots, \sigma'_{N'}(\delta_M)$ are pairwise interior disjoint,
	there may be at most $\OO(N+N')$ pairs of indices 
	\[ 
		\Pi = \{ (i,j)\in [N]\times [N'] \mid \text{$\sigma_i(\delta_M)$ and $\sigma'_j(\delta_M)$ intersect}\}.
	\]
	Therefore, it suffices to compute those pairs $\Pi$ and, 
	for each $(i,j)\in \Pi$ consider the 4 equations $c_i(\delta)=c'_j(\delta)$ with $c_i\in \{a_i,b_i\}$
	and $c'_j\in \{a'_i,b'_i\}$. The solutions to those equations
	are roots of a polynomial of degree at most $4^{h(s_i)}$ because of Lemma~\ref{le:sqrt_solutions}.
	
	The computation of $\Pi$ takes $\OO((N+N') \log (N+N'))$ time, and then we have to compute the 
	roots of the resulting $\OO(N+N')$ polynomials of degree $4^{h(s_i)}$.	
\end{proof}

Using the lemma, we compute in $2^{\OO(k)} (N+N')+ \OO((N+N') \log (N+N')) = 2^{\OO(k)} n \log n$ time
the $4^k \OO(N+N')= 2^{\OO(k)} n$ values of $\delta$ where the boundaries of the segments may intersect
and do a parametric search among them to update $\delta_m$ and $\delta_M$. After that, for any $\delta\in(\delta_m,\delta_M)$, the endpoints of the segments in $X(\delta)$ and $X'(\delta)$ 
are sorted in the same way, and we can easily compute $X(\delta)\cap X'(\delta)$.
Note that the endpoints of the resulting segmentation $X(\delta)\cap X'(\delta)$ keep being
described by $h(s_i)$-square root functions because for each endpoint there was an endpoint
in $X(\delta)$ or $X'(\delta)$.

We repeat $t+u-1\le 4$ times the computation of intersection of two segmentations on $s_i$, until we obtain
\[ 
	X_i(\delta)=\bigcap_{\ell=1}^{t} Y_{\ell}(\delta) \bigcap_{\ell=1}^{u} Z_\ell(\delta) = 
		\bigcap_{\ell=1}^{t} \big(X_{\ell}(\delta)\oplus D(0,\delta)\big)\bigcap_{\ell=1}^{u} \big(C_\ell\oplus D(0,\delta)\big)
		\bigcap s_i.
\]
Altogether, we used $k^{\OO(k)}n\log n$ steps.

\paragraph{Summary of \textbf{Step~\ref{step:inside_loop}}.} 
We perform the computation of $X_i(\delta)$ bottom-up along the significant topology tree $\tau'$.
For each node $s_i\in \SSS'$ of $\tau'$ we spend $k^{\OO(k)}n\log n$ time. At end of processing the significant topology tree $\tau'$, we have computed in $k^{\OO(k)}n\log n$ time values $\delta_m<\delta_M$ such that the set $X_r(\delta)$ is either empty, for all $\delta_m<\delta<\delta_M$, or non-empty, for all $\delta_m<\delta<\delta_M$. This is because there are no combinatorial changes for $\delta\in(\delta_m,\delta_M)$. After we compute the set $X_r(\delta)$ for each significant topology subtree of $\tau$, we know that at least one of these sets is empty. If all of the sets $X_r$ were non-empty, 
then $\delta^*$ should be at most $\delta_m$ by continuity, which cannot be the case.

We repeat \textbf{Step~\ref{step:inside_loop}} for each topology tree $\tau$.
After \textbf{Loop} finishes, we return $\delta^*=\delta_M$.
Since there are $k^{\OO(k)}$ different topology trees to consider, and for each topology
tree we spend $k^{\OO(k)} n\log n$ time, the algorithm takes
$k^{\OO(k)} k^{\OO(k)} n\log n = k^{\OO(k)} n\log n$ time in total.

\begin{theorem}
\label{parametric_theorem}
	The optimization problem \connectivity for $k$ line segments and $n-k$ points can be solved
	performing $k^{\OO(k)}n \log n $ operations. 
	Here, an operation may include a number that has a computation tree of depth $\OO(k)$ whose
	internal nodes are additions, subtractions, multiplications, divisions or square root computations
	and whose leaves contain input numbers and a root of a polynomial of degree at most $4^k$
	with coefficients that are obtained from the input numbers using $2^{\OO(k)}$ 
	multiplications, additions and subtractions.
\end{theorem}
\begin{proof}
	The correctness of the algorithm was argued as the algorithm was described. 
    It remains to discuss the depth of computation tree of the numbers being computed through the algorithm.
	The depths of the computation tree of numbers used in the preprocessing, \textbf{Step~\ref{step:build_tree}} 
	and \textbf{Step~\ref{step:Voronoi}} are $\OO(1)$. 
	Numbers in each iteration of \textbf{Loop} are computed independently of the numbers 
	computed in another iteration. 
	The depth of computation trees of numbers used in \textbf{Step~\ref{step:inside_loop}} is $\OO(k)$,
	but in the calls to the decision problem we are using a root of a polynomial of degree $4^k$.
	Therefore, we are using Theorem~\ref{decision_theorem} with an input number
	that is a root of a polynomial of degree $4^k$ that is computed by using Lemma~\ref{le:sqrt_solutions}.
	The result follows.
\end{proof}

Without diving into the time needed for the algebraic operations performed by the algorithm
and trying to optimize them, we obtain the following.

\begin{corollary}
	The optimization problem \connectivity for $k$ line segments and $n-k$ points can be solved
	in $f(k)n \log n $ time for some computable function $f(\cdot)$.
\end{corollary}

\section{Conclusions}
\label{sec:conclusions}

We have shown that the \connectivity problem for $k$ segments and $n-k$ points in the plane
can be solved in $f(k)n\log n$ time, for some computable function $f(\cdot)$.
The precise function $f$ depends on the time to manipulate algebraic numbers.
The decision problem is simpler, while the algorithm for the optimization problem
uses parametric search.

The algorithms can be extended to $\RR^d$, for any fixed dimension $d$, when the uncertain regions
are segments. The main observations are the following:
\begin{itemize}
	\item A MST for a set of points in $\RR^d$ has maximum degree $c_d=2^{\OO(d)}$. Indeed, Claim~\ref{degree} shows
		that any two edges incident to a point need to have angle at least $\pi/3$. This implies that
		the maximum degree of the MST is bounded by the kissing number in dimension $d$, which
		is known to be $c_d=2^{\OO(d)}$ using a simple volume argument.
	\item A MST for a set $\PP$ of $n$ points in $\RR^d$ can be computed in $\OO(dn^2)$ time
		by constructing the complete graph $K_\PP$ explicitly and using a generic algorithm for MST
		in dense graphs. The term $\OO(d)$ is added because it is needed to compute	each distance.
		More efficient algorithms with a time complexity of 
		$\OO(n^{2-\tfrac{2}{\lceil d/2\rceil +1}+\varepsilon})$,
		for any $\varepsilon>0$, are known~\cite{AgarwalES91} for any fixed dimension $d$. (The constant hidden
		in the $O$-notation depends on $\varepsilon$.)
	\item In Lemma~\ref{ell_bound}, we have to consider the components obtained by removing up to $kc_d$ edges of the MST of $\PP$.
	\item The rest of the discussion follows as written. When constructing the Voronoi diagram restricted
		to a segment $s_i$ and all the other geometric constructions, 
		the dimension of the ambient space does not matter. 
		In fact, when considering a segment $s_i$, we could just replace the input points 
		by points that are coplanar with the segment and have the same distances to the line supporting the segment.
\end{itemize}
All together, when $d$ is constant,
we get an algorithm for $k$ uncertain line segments and $n-k$ points in $\RR^d$ that uses
$\OO(n^{2-\tfrac{2}{\lceil d/2\rceil +1}+\varepsilon}) + f(k) n\log n$ time, for some computable function $f(\cdot)$.
Interestingly, when $k$ is constant, the bottleneck of the computation in our algorithm is obtaining the MST; after
that step, we need $O(f(k)n\log n)$ time.
When $d$ is unbounded, we get an exponential dependency on $d$ because the number of components in the MST
that have to be considered is $O(kc_d)$.

\section*{Acknowledgements}
We are grateful to the reviewers for their feedback and suggestions to improve the paper.

Funded in part by the Slovenian Research and Innovation Agency (P1-0297, J1-1693, J1-2452, N1-0218, N1-0285).
Funded in part by the European Union (ERC, KARST, project number 101071836). Views and opinions expressed are however those of the authors only and do not necessarily reflect those of the European Union or the European Research Council. Neither the European Union nor the granting authority can be held responsible for them.

\bibliographystyle{plainurl}
\bibliography{bibfile}
\end{document}